\pdfoutput=1
\documentclass[final,1p,times,authoryear]{elsarticle}

\usepackage{bm} % for bold \bm
\usepackage{mathrsfs} % for \mathscr
\usepackage{algorithm} % for algorithms

\usepackage[shortlabels]{enumitem}
\setlist{topsep=0.25\baselineskip,partopsep=0pt,itemsep=1pt,parsep=0pt}

\usepackage{amsmath,amsfonts,amssymb,amsthm,thmtools}
\declaretheoremstyle[headfont=\normalfont\bfseries,bodyfont=\normalfont]{myremark}
\declaretheorem[style=plain,parent=section]{definition}
\declaretheorem[sibling=definition]{theorem}

\declaretheorem[sibling=definition]{proposition}
\declaretheorem[sibling=definition]{lemma}
\declaretheorem[style=myremark,sibling=definition,qed={\qedsymbol}]{remark}

\usepackage[linkcolor=black,colorlinks=true,citecolor=blue,urlcolor=blue]{hyperref} 
\usepackage[capitalise]{cleveref}

\newcommand{\storeArg}{} % aux, not to be used in document!!

\newcommand{\NN}{\mathbb{N}} % nonnegative integers
\newcommand{\var}{T} % variable for univariate polynomials
\newcommand{\field}{\mathbb{K}} % base field
\newcommand{\polRing}{\field[\var]} % polynomial ring

\newcommand{\Poxi}{[\mkern-3mu[ \var^{-1} ]\mkern-3.2mu]}

\newcommand{\matSpace}[1][\rdim]{\renewcommand\storeArg{#1}\matSpaceAux} % polynomial matrix space, 2 opt args
\newcommand{\matSpaceAux}[1][\storeArg]{\field^{\storeArg \times #1}} % not to be used in document
\newcommand{\polMatSpace}[1][\rdim]{\renewcommand\storeArg{#1}\polMatSpaceAux} % polynomial matrix space, 2 opt args
\newcommand{\polMatSpaceAux}[1][\storeArg]{\polRing^{\storeArg \times #1}} % not to be used in document
\newcommand{\mat}[1]{\bm{\MakeUppercase{#1}}} % for a matrix
\newcommand{\row}[1]{\bm{\MakeLowercase{#1}}} % for a matrix
\newcommand{\col}[1]{\bm{\MakeLowercase{#1}}} % for a matrix
\newcommand{\rdim}{m} % row dimension
\newcommand{\cdim}{{m'}} % column dimension
\newcommand{\diag}[1]{\mathrm{Diag}(#1)}  % diagonal matrix with diagonal entries #1
\newcommand{\seqelt}[1]{\bm{F}_{#1}} % element of sequence of matrices
\newcommand{\sseqeltSpace}{\matSpace[\rdim][\rdim]} % element of sequence of square matrices
\newcommand{\seq}{\mat{\mathcal{F}}} % sequence of matrices
\newcommand{\seqL}{\mat{\mathcal{L}}} % sequence of matrices
\newcommand{\seqpm}{\mat{Z}} % power series matrix from a sequence
\newcommand{\rel}{\col{p}} % linear relation
\newcommand{\relbas}{\mat{P}} % linear relation
\newcommand{\relSpace}{\polMatSpace[1][\rdim]} % space for linear relations
\newcommand{\relbasSpace}{\polMatSpace[\rdim][\rdim]} % space for linear relations
\newcommand{\num}{\row{q}} % numerator for linear recurrence relation
\newcommand{\nummat}{\mat{Q}} % numerator for linear recurrence relation basis
\newcommand{\rem}{\row{r}} % remnant for linear recurrence relation
\newcommand{\remmat}{\mat{R}} % remnant for linear recurrence relation basis
\newcommand{\remSpace}{\polMatSpace[1][\cdim]} % space for linear relations
\newcommand{\degBd}{d} % bound on degree of minimal generator
\newcommand{\degBdr}{d_{r}} % bound on degree of a right minimal generator
\newcommand{\degBdl}{d_{\ell}} % bound on degree of a left minimal generator

\newcommand{\cdeg}[2][]{\mathrm{cdeg}_{{#1}}(#2)} % shifted column degree
\newcommand{\sys}{\mat{F}} % input matrix series to approximant basis
\newcommand{\appMod}[2]{\mathcal{A}(#1,#2)} % module of approximants for #2 at order #1
\newcommand{\basis}{\mathscr{B}}
\newcommand{\trace}{\operatorname{trace}}
\newcommand{\softO}[1]{O{\tilde{~}}(#1)} % module of approximants for #2 at order #1
\newcommand{\genseries}{Z}
\newcommand{\minpoly}{P}
\newcommand{\mainalgoname}{\mathsf{ BlockParametrization}}
\newcommand{\lf}{X}
\newcommand{\mf}{Y}
\newcommand{\residueI}{\mathscr{Q}}
\newcommand{\sqfree}{Q}
\newcommand{\trsp}[1]{#1^{\mathsf{T}}} % transpose of a matrix
\newcommand{\itrsp}[1]{#1^{\mathsf{-T}}} % inverse transpose of a matrix

\newcommand{\density}{\rho}

\def\M {\ensuremath{\mathsf{M}}}
\def\PP {\ensuremath{\mathsf{P}}}

\def\dg{\kappa}

\def\Q {\ensuremath{\mathbb{Q}}}
\def\N {\ensuremath{\mathbb{N}}}

\def\F {\ensuremath{\mathbb{F}}}

\def\K{\mathbb{K}}
\def\K {\ensuremath{\mathbb{K}}}
\def\Kbar {{\ensuremath{\overline{\mathbb{K}}}}}

\def\D {\ensuremath{D}}
\def\m {\ensuremath{\mathfrak{m}}}

\DeclareBoldMathCommand{\bell}{\ell}
\DeclareBoldMathCommand{\be}{e}
\DeclareBoldMathCommand{\bu}{u}
\DeclareBoldMathCommand{\bv}{v}
\DeclareBoldMathCommand{\bX}{X}
\DeclareBoldMathCommand{\bx}{x}
\DeclareBoldMathCommand{\balpha}{\alpha}
\DeclareBoldMathCommand{\bbeta}{\beta}
\DeclareBoldMathCommand{\mA}{A}
\DeclareBoldMathCommand{\mB}{B}
\DeclareBoldMathCommand{\mD}{D}
\DeclareBoldMathCommand{\mF}{F}
\DeclareBoldMathCommand{\mG}{G}
\DeclareBoldMathCommand{\mI}{I}
\DeclareBoldMathCommand{\mM}{M}
\DeclareBoldMathCommand{\mNs}{N^*}
\DeclareBoldMathCommand{\mN}{N}
\DeclareBoldMathCommand{\mS}{S}
\DeclareBoldMathCommand{\mT}{T}
\DeclareBoldMathCommand{\mU}{U}
\DeclareBoldMathCommand{\mV}{V}
\DeclareBoldMathCommand{\mW}{W}
\DeclareBoldMathCommand{\mX}{X}
\DeclareBoldMathCommand{\mY}{Y}
\DeclareBoldMathCommand{\mZ}{Z}
\DeclareBoldMathCommand{\bell}{\ell}
\DeclareBoldMathCommand{\be}{e}
\DeclareBoldMathCommand{\bu}{u}
\DeclareBoldMathCommand{\bv}{v}
\DeclareBoldMathCommand{\bX}{X}
\DeclareBoldMathCommand{\bx}{x}
\DeclareBoldMathCommand{\balpha}{\alpha}
\DeclareBoldMathCommand{\bbeta}{\beta}

\newcommand{\mUt}{\trsp{\mU}}

\begin{document}

\begin{frontmatter}

  \title{Block-Krylov techniques in the context of \\ sparse-FGLM algorithms}

  \author{Seung Gyu Hyun}
  \address{Cheriton School of Computer Science, University of Waterloo}

  \author{Vincent Neiger}
  \address{Univ.~Limoges, CNRS, XLIM, UMR 7252, F-87000 Limoges, France}

  \author{Hamid Rahkooy}
  \address{Cheriton School of Computer Science, University of Waterloo}

  \author{\'Eric Schost}
  \address{Cheriton School of Computer Science, University of Waterloo}

  \begin{abstract}
    Consider a zero-dimensional ideal $I$ in $\K[X_1,\dots,X_n]$.  Inspired by
    Faug\`ere and Mou's Sparse FGLM algorithm,
    we use Krylov sequences based on multiplication matrices of $I$ in order to
    compute a description of its zero set by means of univariate polynomials.

    Steel recently showed how to use Coppersmith's block-Wiedemann algorithm in
    this context; he describes an algorithm that can be easily parallelized, but
    only computes parts of the output in this manner. Using generating series
    expressions going back to work of Bostan, Salvy, and Schost, we show how to
    compute the entire output for a small overhead, without making any assumption
    on the ideal $I$ other than it having dimension zero. We then propose a refinement of this idea that partially
    avoids the introduction of a generic linear form.  We comment on experimental
    results obtained by an implementation based on the C++ libraries Eigen, LinBox and
    NTL.
  \end{abstract}


\end{frontmatter}

\section{Introduction}

Computing the Gr\"obner basis of an ideal with respect to a given term
ordering is an essential step in solving systems of polynomials.
Certain term orderings, such as the degree reverse lexicographic
ordering (\textit{degrevlex}), tend to make the computation of the
Gr\"obner basis faster. This has been observed empirically since the
1980's and is now supported by theoretical results, at least for some
``nice'' families of inputs, such as complete intersections or certain
determinantal systems~\citep{Faugere02,FaSaSp13,BaFaSa15}.  On the
other hand, other orderings, such as the lexicographic ordering
(\textit{lex}), make it easier to find the coordinates of the
solutions, or to perform arithmetic operations in the corresponding
residue class ring.  For instance, for a zero-dimensional radical
ideal $I$ in generic coordinates in $\K[X_1,\dots,X_n]$, for some
field $\K$, the Gr\"obner basis of $I$ for the lexicographic ordering
with $X_1 > \cdots > X_n$ has the form
\begin{equation}\label{eq:shapelemma}
  \{ X_1 - R_1(X_n),\dots,X_{n-1}-R_{n-1}(X_n),R_n(X_n)\},
\end{equation}
with all $R_i$'s, for $i =1,\dots,n-1$, of degree less than
$\deg(R_n)$ (and $R_n$ squarefree); this is known as the {\em
shape lemma}~\citep{GiMo89}. The points in the variety $V(I) \subset
\Kbar{}^n$ are then
$$\{ ( R_1(\tau), \dots, R_{n-1}(\tau), \tau ) \mid \tau \in \Kbar
\;\,\text{is a root of}\;\, R_n\}.$$ As a result, the standard approach to
solve a zero-dimensional system by means of Gr\"obner basis
algorithms is to first compute a Gr\"obner basis for a degree ordering
and then convert it to a more exploitable output, such as a lexicographic
basis. As pointed out in~\citep{FaMo17}, the latter step, while of
polynomial complexity, can now be a bottleneck in practice. This
paper will thus focus on this step; in order to describe our 
contributions, we first discuss previous work on the question.

Let $I$ be a zero-dimensional ideal in $\K[X_1,\dots,X_n]$. As
input, we assume that we know a monomial basis $\basis$ of
$\residueI=\K[X_1,\dots,X_n]/I$ together with the multiplication
matrices $\mM_1,\dots,\mM_n$ of respectively $X_1,\dots,X_n$ in this
basis. We denote by $D$ the degree of $I$, which is the vector space
dimension of $\residueI$. We should stress that starting from a
degree Gr\"obner basis of $I$, computing the multiplication matrices
efficiently is not a straightforward task. \citep{FaGiLaMo93}
showed how to do it in time $O(nD^3)$; more
recently, algorithms have been given with cost bound
$\softO{nD^\omega}$~\citep{FaGaHuRe13,FaGaHuRe14,Neiger16}, at least
for some favorable families of inputs. Here, the notation
$O\tilde{~}$ hides polylogarithmic factors and $\omega$ is a feasible
exponent for matrix multiplication over a ring (commutative, with
$1$). While improving these results is an interesting question in
itself, we will not address it in this paper.

Given such an input, including the multiplication matrices, the FGLM
algorithm~\citep{FaGiLaMo93} computes the lexicographic Gr\"obner basis
of $I$ in $O(nD^3)$ operations in $\K$.  While the algorithm has an
obvious relation to linear algebra, lowering the runtime to
$O\tilde{~}(nD^\omega)$ was only recently
achieved~\citep{FaGaHuRe13,FaGaHuRe14,Neiger16}.

Polynomials as in~\cref{eq:shapelemma} form a very useful data
structure, but there is no guarantee that the lexicographic Gr\"obner
basis of $I$ has such a shape (even in generic coordinates). When it does, we will say that $I$ is
in {\em shape position}; some sufficient conditions for being in shape
position are detailed in~\citep{BeMoMaTr94}.  As an alternative, one
may then use the Rational Univariate Representation algorithm of
\citet{Rouillier99} (see also~\citep{AlBeRoWo94,BeWo96} for related
considerations). The output is a description of the zero-set $V(I)$ by
means of univariate rational functions
\begin{equation}\label{eq:RUR}
  \left\{  F(T)=0, \quad X_1 = \frac{G_1(T)}{G(T)}, \dots,X_n = \frac{G_n(T)}{G(T)} \right\},
\end{equation}
where the multiplicity of a root $\tau$ of $F$ coincides with that of
$I$ at the corresponding point
$(G_1(\tau)/G(\tau),\dots,G_n(\tau)/G(\tau)) \in V(I)$. The fact that
we use rational functions makes it possible to control
precisely the bit-size of their coefficients, if working over $\K=\Q$.

The algorithms of~\citet{AlBeRoWo94, BeWo96, Rouillier99} rely on
\emph{duality}, which will be at the core of our algorithms as well.
Indeed, these algorithms compute sequences of values of the form
$v_s=(\trace(\lf^s))_{s \ge 0}$ and 
$v_{s,i}=(\trace(\lf^s X_i))_{s \ge 0}$, where $\trace: \residueI \to \K$ is the trace 
form and $\lf=t_1 X_1 + \cdots + t_n X_n$ is a generic $\K$-linear combination of the variables.
From these values, one may then recover the output in~\cref{eq:RUR} by means
of structured linear algebra calculations.

A drawback of this approach is that we need to know the trace of all
elements of the basis $\basis$; while feasible in polynomial time,
this is by no means straightforward. \citet{BoSaSc03}
introduced randomization to alleviate this issue. They show
that computing values such as $\ell(\lf^s)$ and $\ell(\lf^s X_i)$, where
$\lf$ is as above and 
$\ell$ is a random $\K$-linear form $\residueI \to \K$, allows one to deduce
a description of $V(I)$ of the form
\begin{equation}\label{eq:BoSaSc03}
  \{  \sqfree(T)=0, \quad X_1 = V_1(T), \dots,X_n = V_n(T) \},
\end{equation}
where $\sqfree$ is a monic squarefree polynomial in $\K[T]$ and $V_i$ is in
$\K[T]$ of degree less than $\deg(\sqfree)$ for all $i$. The tuple
$((\sqfree,V_1,\dots,V_n),\lf)$ computed by  such algorithms
will be called a {\em zero-dimensional parametrization} of $V(I)$. In
particular, it generally differs from the description
in~\cref{eq:RUR}, since the latter keeps track of the multiplicities
of the solutions (the algorithm in~\citep{BoSaSc03} actually computes
the {\em nil-indices} of the solutions). Remark that starting from
such an output, we can reconstruct the local structure of $I$ at 
its roots, using algorithms from~\citep{MaMoMo96},~\citep{Mourrain97} or~\citep{NeRaSc17}.

The most costly part of the algorithm of~\citep{BoSaSc03} is the
computation of the values $\ell(\lf^s)$ and $\ell(\lf^s X_i)$; the
rest essentially boils down to applying the Berlekamp-Massey algorithm
and univariate polynomial arithmetic. \citet{FaMo17} pointed out that the
multiplication matrices
$\mM_1,\dots,\mM_n$ can be expected to be sparse; for generic inputs,
they gave precise estimates on the sparsity of these matrices,
assuming the validity of a conjecture by
\citet{MorenoSocias91}.  On this basis, they designed
several forms of {\em sparse FGLM} algorithms. For instance, if $I$ is
in shape position, the algorithms in \citep{FaMo17} recover its
lexicographic basis, which is as in~\eqref{eq:shapelemma}, by also
considering values of a linear form $\ell:\residueI \to \K$. For less
favorable inputs, these algorithms fall back either on the
algorithm of \citet{Sakata90} or on plain FGLM.

The ideas at play in these algorithms are essentially based on Krylov subspace
methods, using projections as well as Berlekamp-Massey techniques, along the
lines of the algorithm of \citet{Wiedemann86} to solve sparse linear systems.
These techniques have also been widely used in integer
factorization or discrete logarithm calculations, going back to~\citep{LaOd90}.
It has become customary to design block versions of such algorithms to
parallelize their bottleneck, as pioneered by \citet{Coppersmith94}
in the context of integer factorization: it is then natural to adapt this
strategy to our situation. This was already discussed by
\citet{Steel15}, where he showed how to compute the analogue of the
polynomial $\sqfree$ in \cref{eq:BoSaSc03} using such techniques. In that
reference, one is only interested in the solutions in the base field $\K$ ($\K$
being a finite field in that context): the algorithm computes the roots of
$\sqfree$ in $\K$ and substitutes them in the input system, before computing a
Gr\"obner basis in $n-1$ variables for each of them.

Our first contribution is to give a block version of the algorithm
in~\citep{BoSaSc03} that extends the approach introduced
in~\citep{Steel15} to compute all polynomials in~\cref{eq:BoSaSc03} for
essentially the same cost as the computation of $\sqfree$. More
precisely, the bottleneck of the algorithm of~\citep{Steel15} is the
computation of a block-Krylov sequence; we show that once this
sequence has been computed, not only $\sqfree$ but also all other
polynomials in the zero-dimensional parametrization can be efficiently
obtained. Compared with the algorithms of~\citep{FaMo17}, a notable
difference is that our algorithm deals with any  zero-dimensional ideal
$I$ (for instance, we do not require shape position), but the
base field must have sufficiently large characteristic (and our output is
somewhat weaker than a Gr\"obner basis, since multiplicities are not
computed). While we focus on the case where the multiplication
matrices are sparse, we also give a cost analysis for the case of
dense multiplication matrices.

Our second contribution is a refinement of our first algorithm, where
we try to avoid computations with a generic linear form $\lf =t_1 X_1 +
\cdots + t_n X_n$ to the extent possible (this is motivated by the
fact that the multiplication matrix of $\lf$ is often denser than those
of the variables $X_i$). The algorithm first computes a zero-dimensional
parametrization of a subset of $V(I)$ for which we can take $\lf$
equal to (say) $X_1$, and falls back on the previous approach for the
residual part; if the former set has large cardinality, this is
expected to provide a speed-up over the general algorithm.

For experimental purposes, our algorithms have been implemented in C++ using
the libraries Eigen~\citep{Eigen}, LinBox~\citep{LinBox} and NTL~\citep{NTL}.

The paper is organized as follows. The next section mainly reviews
known results on scalar and matrix recurrent sequences, and introduces
a simple useful algorithm to compute a so-called {\em scalar
  numerator} for such sequences. \cref{sec:seq0} describes sequences
that arise in the context of FGLM-like algorithms; we prove slightly
refined versions of results from~\citep{BoSaSc03} that will be used
throughout the paper. The main algorithm is given in \cref{sec:main},
and the refinement mentioned above is in \cref{sec:original}.
Finally, in appendix, we prove a few technical statements on linearly
recurrent matrix sequences.

\paragraph{Complexity model}
We measure the cost of our algorithms by counting basic operations in
$\K$ at unit cost. Most algorithms are randomized; they involve the
choice of a vector $\gamma \in \K^S$ of field elements, for an integer
$S$ that depends on the size of our input, and success is guaranteed
if the vector $\gamma$ avoids a hypersurface of the parameter space
$\K^S$.

Suppose the input ideal $I$ is generated by polynomials
$F_1,\dots,F_t$.  Given a zero-dimensional parametrization
$((\sqfree,V_1,\dots,V_n),\lf)$ found by our algorithms, one can
always evaluate $F_1,\dots,F_t$ at $X_1 =V_1,\dots,X_n=V_n$, doing all
computations modulo $\sqfree$. This allows us to verify whether the
output describes a subset of $V(F_1,\dots,F_t)$, but not whether we
have found all solutions.  If $\deg(Q)$ coincides with the dimension
of $\residueI=\K[X_1,\dots,X_n]/I$, we can infer that we have all
solutions (and that $I$ is radical), so our output is correct.

In what follows, we assume $\omega>2$, in order to simplify a few cost
estimates. We use a time function $d \mapsto \M(d)$ for the cost of
univariate polynomial multiplication over $\K$, for which we assume
the super-linearity properties of \citep[Section~8.4]{GaGe13}.  Then,
two $m\times m$ matrices over $\K[T]$ whose degree is less than $d$
can be multiplied using $O(m^\omega \M(d))$ operations in $\K$.

\paragraph{Acknowledgments} We wish to thank Chenqi Mou, Dave Saunders, and
Gilles Villard for several discussions, and a reviewer of the first
version of this paper for their useful remarks. This research was
partially supported by NSERC (Schost's NSERC Discovery Grant) and by
the CNRS-INS2I Institute through its program for young researchers
(Neiger's project ARCADIE).

%%%%%%%%%%%%%%%%%%%%%%%%%%%%%%%%%%%%%%%%%%%%%%%%%%%%%%%%%%%%
%%%%%%%%%%%%%%%%%%%%%%%%%%%%%%%%%%%%%%%%%%%%%%%%%%%%%%%%%%%%
%%%%%%%%%%%%%%%%%%%%%%%%%%%%%%%%%%%%%%%%%%%%%%%%%%%%%%%%%%%%

\section{Linearly recurrent sequences}

This section starts with a review of known facts on linearly recurrent
sequences: we first discuss the scalar case, and then we show how the ideas
carry over to matrix sequences. 
The main results we will need from the first three subsections
are~\cref{coro:cost_approx} and \cref{randXY}, which give cost
estimates for the computation of a {\em minimal matrix generator} of a
linearly recurrent matrix sequence, as well as a degree bound for such
generators, when the matrices we consider are obtained from a Krylov
sequence. These results are for the most part not new
(see~\citep{Villard97,Villard97a,KalVil01,Turner02,KaVi04}), but the
cost analysis we give uses results not available when those references
were written.  The fourth and last subsection presents a useful result
for our main algorithm that allows us to compute a ``numerator'' for a
scalar sequence from a similar object obtained for a matrix sequence.

%%%%%%%%%%%%%%%%%%%%%%%%%%%%%%%%%%%%%%%%%%%%%%%%%%%%%%%%%%%%

\subsection{Scalar sequences} \label{section:linseq}

Let $\K$ be a field and consider a sequence $\mathcal{L}=(\ell_s)_{s
\ge 0} \in \K^\N$. We say that a degree $d$ polynomial $\minpoly =
p_0 + \cdots + p_d T^d \in\K[T]$ {\em cancels} the sequence
$\mathcal{L}$ if $p_0 \ell_s + \cdots + p_d \ell_{s+d}=0$ for all $s
\ge 0$. The sequence $\mathcal{L}$ is {\em linearly recurrent} if
there exists a nonzero polynomial that cancels it.  The {\em minimal
polynomial} of a linearly recurrent sequence
$\mathcal{L}=(\ell_s)_{s \ge 0}$ is the monic polynomial of lowest
degree that cancels it; the {\em order} of $\mathcal{L}$ is the degree
of this polynomial~$\minpoly$.

In terms of generating series, it is easier to work here with
generating series in the variable $1/T$.  Let thus $\genseries =
\sum_{s\ge0} \ell_s / T^{s+1}$ and $\minpoly$ be any polynomial; then,
$\minpoly$ cancels the sequence $\mathcal{L}$ if and only if
$Q=\minpoly \genseries $ is a polynomial, in which case $Q$ must have
degree less than $\deg(\minpoly)$.  In particular, given a scalar
sequence and a polynomial that cancels it, there is a well-defined
notion of associated numerator. This is formalized in the next
definition.

\begin{definition}
  \label{def:omega}
  Let $\mathcal{L}=(\ell_s)_{s \ge 0}\in \K^\N$ be a sequence and $P$ be a
  polynomial that cancels $\mathcal{L}$. Then, the {\em numerator} of $\mathcal{L}$
  with respect to $P$ is denoted by $\Omega(\mathcal{L},P)$ and defined as 
  \[
    \Omega(\mathcal{L},P) = P \genseries, \quad\text{where}\quad
    \genseries=\sum_{s \ge 0} \frac {\ell_s}{T^{s+1}}.
  \]
  In particular, $\Omega(\mathcal{L},P)$ is a polynomial of
  degree less than $\deg(P)$.
\end{definition}

The numerators thus defined are related to the Lagrange interpolation
polynomials; this will explain why they play an important role in our
main algorithm. Assuming that $\mathcal{L}$ is known to have order at
most $d$ and that we are given $\ell_0,\dots,\ell_{2d-1}$, we can recover
its minimal polynomial $P$ by means of the Berlekamp-Massey algorithm,
or of Euclid's algorithm~\citep{BrGuYu80}. Given $P$ and
$\ell_0,\dots,\ell_{d-1}$, we can deduce
$\Omega(\mathcal{L},P)$ by a simple multiplication.

As an example, consider the Fibonacci sequence $\mathcal{L} =
(1,1,2,3,5,8,\dots)$, which is linearly recurrent with minimal
polynomial $P=T^2-T-1$.
Defining
$\genseries = \sum_{s\ge0} \ell_{s}/T^{s+1}$, we obtain
\[
  \Omega(\mathcal{L},P) = (T^2-T-1)\left (\frac 1T +\frac 1{T^2} + \frac
  2{T^3} + \frac 3{T^4} + \cdots \right ) =T.
\]

%%%%%%%%%%%%%%%%%%%%%%%%%%%%%%%%%%%%%%%%%%%%%%%%%%%%%%%%%%%%

\subsection{Linearly recurrent matrix sequences}\label{section:matrix_seq}

Next, we discuss the analogue of these ideas for matrix sequences; our
main goal is to give a cost estimate for the computation of a suitable
{\em matrix generator} for a matrix sequence, obtained by means of
recent algorithms for approximant bases. A similar discussion, relying
on the approximant basis algorithm in~\citep{BecLab94}, can be found
in~\citep[Chapter~4]{Turner02}; as a result, in the body of the text,
we indicate only the key definitions and results (and in particular,
an improved complexity analysis). Proofs of some statements are in
appendix.

We first define linear recurrences for matrix sequences over a field
$\field$ as in~\citep[Section~3]{KalVil01}
or~\citep[Definition~4.2]{Turner02}, hereby extending the above
notions for scalar sequences. We only discuss the square case, since
this is what we need below; the whole discussion can be generalized to
rectangular matrices.
\begin{definition}
  \label{dfn:recurrence_relation}
  Let $\seq = (\seqelt{s})_{s\ge 0}$ be a sequence of matrices in
  $\sseqeltSpace$. Then,
  \begin{itemize}
    \item a polynomial vector $\rel = \row{p}_0 + \cdots +
      \row{p}_\degBd T^\degBd \in
      \relSpace$ is a \emph{left vector relation for $\seq$} if $
      \row{p}_0 \seqelt{s} + \cdots + \row{p}_{\degBd}
      \seqelt{s+\degBd} = \row{0}$ holds for all $s \ge 0$;
    \item the sequence $\seq$ is \emph{linearly recurrent} if the set
      of left vector relations for $\seq$ is a $\polRing$-submodule of
      $\relSpace$ of rank $\rdim$.
  \end{itemize}
\end{definition}
Note that the set of left vector relations is always a free
$\K[T]$-module, but its rank may be less than $\rdim$. The following
equivalent characterization of linearly recurrent sequences can be
found for example in \citep{Villard97,KalVil01,Turner02}.
\begin{lemma}
  \label{lem:module_rank}
  Let $\seq = (\seqelt{s})_{s\ge 0}$ be a sequence of matrices in
  $\sseqeltSpace$. Then $\seq$ is linearly recurrent if and only if it admits
  a nonzero scalar relation, that is, if there exists a non-zero
  polynomial $P = p_0 + \cdots + p_\degBd T^\degBd \in \polRing$ such
  that the identity $p_0 \seqelt{s} + \cdots + p_{\degBd}
  \seqelt{s+\degBd} = \mat{0}$ holds for all $s \ge 0$.
\end{lemma}

\begin{definition}
  \label{dfn:matrix_generator}
  Let $\seq = (\seqelt{s})_{s\ge 0}$ be a linearly recurrent sequence of matrices in
  $\sseqeltSpace$. A matrix 
  $\mat{P}$ in $\relbasSpace$ is 
  \begin{itemize}
  \item a \emph{left matrix relation} if its rows are all left vector relations for $\seq$;
  \item a \emph{left matrix generator} if its rows form a basis of the
    module of left vector relations for $\seq$.
  \end{itemize}
  In the latter case, $\mat{P}$ is a {\em minimal generator} if it is in row
  reduced form.
\end{definition}

We refer to \citep{Wolovich74,Kailath80} for a definition of reduced forms.
The following proposition, which can be proved by a direct examination
of the terms in the product $\relbas \seqpm$, shows that matrix
relations are denominators in a fraction description of the generating
series of the sequence. As suggested in the previous subsection,
working with generating series in $1/T$ turns out to be the most
convenient choice here.

\begin{proposition}
  Let $\seq = (\seqelt{s})_{s\ge 0}$ be a sequence of matrices in
  $\sseqeltSpace$ and let $\mat{P}$ be a nonsingular matrix in $\relbasSpace$. 
  Then, $\mat{P}$ is a left matrix relation for $\seq$ if and only
  if the generating series $\seqpm = \sum_{s\ge 0} \seqelt{s} /
  \var^{s+1} \in \field\Poxi^{\rdim \times \rdim}$ can be written as a
  matrix fraction $\seqpm = \relbas^{-1} \nummat$, with $\nummat \in
  \polMatSpace[\rdim][\rdim]$. In this case, 
  the \(i\)th row of $\nummat$ has degree less than the \(i\)th row of
  $\relbas$, for \(1\le i\le \rdim\).
\end{proposition}

Given a nonsingular matrix relation $\relbas$ for $\seq$, we will
thus write $\mat{\Omega}(\seq, \relbas)= \relbas \seqpm \in
\polMatSpace[\rdim][\rdim]$, generalizing \cref{def:omega}. 
By the previous proposition, this is a polynomial matrix, whose $i$th
row has degree less than the $i$th row of $\mat{P}$ for $1\le
i\le\rdim$.  As in the scalar case, if $\mat{P}$ has degree $d$, we
only need to know $\seqelt{0},\dots,\seqelt{d-1}$ to recover
$\mat{\Omega}(\seq, \relbas)$ as
\begin{align}\label{eq:computeOmega}
  \mat{\Omega}(\seq, \relbas) =  \left ( \relbas \cdot \sum_{s=0}^{d-1} 
    \seqelt{d-1-s} T^s \right ) \;\mathrm{div}\; T^d,
\end{align}
where ``$\mathrm{div}\; T^d$'' means we keep the quotient of each
entry by $T^d$.  Given $\mat{P}$ and $\seqelt{0},\dots,\seqelt{d-1}$, the cost of computing $\mat{\Omega}(\seq, \relbas) $ is then $O(\rdim^\omega \M(d))$
operations in $\K$.

Given a linearly recurrent sequence, our goal will be to find a minimal left
matrix generator of it (from which we will easily deduce the corresponding
numerator). In order to do this, we will assume that we know bounds on the
degree of this left generator, but also on that of a {\em right} generator.
Indeed, in all the previous discussion, we may also consider vector relations
operating on the right. In particular, \cref{lem:module_rank} shows that if a
sequence is linearly recurrent, these right relations form a submodule of
$\polMatSpace[\rdim][1]$ of rank $\rdim$, so that a linearly recurrent sequence
also admits right generators.

As it turns out, all minimal left generators have the same degree (by property
of reduced forms); the same remark holds for minimal right generators. Knowing
bounds $(\degBdl,\degBdr)$ on these degrees allows us to control the number of
terms of the sequence we will access during the algorithm (the bounds
$(\degBdl,\degBdr)$ correspond to $(\gamma_1,\gamma_2)$ in
\citep[Definitions~4.6~and~4.7]{Turner02} and $(\delta_l,\delta_r)$ in
\citep[Section~4.2]{Villard97a}).  In concrete terms, the fast computation of
minimal matrix generators is usually handled via minimal approximant basis
algorithms \citep[see for example][]{Villard97,Turner02,GioLeb14}. The runtime
below is obtained by using the divide and conquer approximant basis algorithm
in~\citep{GiJeVi03}.

\begin{theorem}\label{coro:cost_approx}
  Let $\seq = (\seqelt{s})_{s\ge 0}$ be a linearly recurrent sequence
  of matrices in $\sseqeltSpace$ and let $\degBd = \degBdl+\degBdr+1$,
  where $(\degBdl,\degBdr) \in \NN^2$ are such that the minimal left
  (resp.~right) matrix generators of $\seq$ have degree at most $\degBdl$
  (resp.~at most $\degBdr$). Then, given $\seqelt{0},\dots,\seqelt{d-1}$,
  one can compute a minimal left matrix generator of $\seq$ in
  $O(\rdim^\omega \M(\degBd) \log(\degBd))$ operations in $\field$.
\end{theorem}

%%%%%%%%%%%%%%%%%%%%%%%%%%%%%%%%%%%%%%%%%%%%%%%%%%%%%%%%%%%%

\subsection{Application to the block Wiedemann algorithm}\label{ssec:appliW}

We next apply the results seen above to a particular class of matrix
sequences, namely the Krylov sequences used in Coppersmith's block
Wiedemann algorithm~\citep{Coppersmith94}.  Let $\mM$ be in
$\mathbb{K}^{D \times D}$ and $\mU,\mV \in \mathbb{K}^{D \times m}$ be
two blocking matrices for some $m\le D$. We can then define the Krylov
sequence $\seq_{\mU,\mV}=(\seqelt{s,\mU,\mV})_{s \ge 0} \subset
\matSpace[m]$ by
$$\seqelt{s,\mU,\mV} = \mUt \mM^s \mV, \quad s \ge 0.$$ This
sequence is linearly recurrent, since the minimal polynomial of $\mM$
is a scalar relation for it. The following theorem states some useful
properties of any minimal left generator of $\seq_{\mU,\mV}$, with in
particular a bound on its degree, for generic choices of $\mU$ and
$\mV$; we also state properties of the invariant factors of such a
generator.  These results are not new, as all statements can be
found in~\citep{Villard97a} and~\citep{KaVi04} (see also~\citep{Kaltofen95}
for an analysis of the block Wiedemann algorithm). 

We let $s_1, \dots, s_r$ be the nontrivial invariant factors of $T\mI_D
- \mM$, ordered such that $s_i$ divides $s_{i-1}$ for \(2 \le i \le r\), and
let $d_i = $ deg$(s_i)$ for all $i$; for $i > r$, we let $s_i$ = 1,
with $d_i = 0$.  We define $\nu = d_1 + \cdots + d_m \le D$ and
$\delta = \lceil \nu / m \rceil \le \lceil D / m \rceil$.

\begin{theorem}
  \label{randXY}
  For a generic choice of $\mU$ and $\mV$ in $\K^{D \times m}$, the
  following holds.  Let $\mat{P}_{\mU,\mV}$ be a minimal left
  generator for $\seq_{\mU,\mV}$ and denote by $\sigma_1, \dots,
  \sigma_k$ the invariant factors of $\mat{P}_{\mU,\mV}$, for some $k
  \le m$, ordered as above, and write $\sigma_{k+1}=\cdots=\sigma_m=1$.
  Then,
  \begin{itemize}
    \item $\mat{P}_{\mU,\mV}$ is a minimal left generator for the
      matrix sequence $\seqL_{\mU} = (\mUt \mM^s)_{s \ge 0}$;
    \item $\mat{P}_{\mU,\mV}$ has degree $\delta$;
    \item $s_i = \sigma_i$ for $1 \le i \le m$.
  \end{itemize}
\end{theorem}

In particular, combining \cref{coro:cost_approx} and \cref{randXY}, we
deduce that for generic $\mU,\mV$, given the first $2 \lceil D/m
\rceil+1$ terms of the sequence $\seq_{\mU,\mV}$, we can recover a
minimal matrix generator  $\mat{P}_{\mU,\mV}$ of it using $O(m^{\omega-1} \M(D) \log(D))
\subset \softO{m^{\omega-1} D}$ operations in $\K$.

Besides, the theorem shows that for generic $\mU$ and $\mV$, the
largest invariant factor $\sigma_1$ of $\mat{P}_{\mU,\mV}$ is the
minimal polynomial $\minpoly=s_1$ of $\mM$.  Given $\mat{P}_{\mU,\mV}$,
$\minpoly$ can thus be computed by solving a linear system
$\mat{P}_{\mU,\mV} \, \col{x} = \col{y}$, where $\col{y}$ is a vector of
$m$ randomly chosen elements in $\K$: for a generic choice of
$\col{y}$, the least common multiple of the denominators of the
entries of $\col{x}$ is $\minpoly$.  Thus, both $\minpoly$ and $\col{x}$
can be computed using
high-order lifting \citep[Algorithm~5]{Stor03} on input
$\mat{P}_{\mU,\mV}$ and $\col{y}$; by
\citep[Corollary~16]{Stor03}, this costs
\begin{equation}\label{eqn:hol_cost}
  O(m^{\omega-1} \M(D) \log(D)\log(m)) \subset \softO{m^{\omega-1}D}
  \end{equation}
operations in $\K$.  The latter algorithm is randomized, since it
chooses a random point in $\K$ to compute (parts of) the expansion of
the inverse of $\mat{P}_{\mU,\mV}$. An alternative solution would be
to compute the Smith form of $\mat{P}_{\mU,\mV}$ using an algorithm
such as that in \citep[Section~17]{Stor03}, yet the cost would be
slightly higher on the level of logarithmic factors.

%%%%%%%%%%%%%%%%%%%%%%%%%%%%%%%%%%%%%%%%%%%%%%%%%%%%%%%%%%%%

\subsection{Computing a scalar numerator}\label{ssec:scalar_numer}

Let us keep the notation of the previous subsection.  The main
advantage of using the block Wiedemann algorithm is that it allows one
to distribute the bulk of the computation in a straightforward manner:
on a platform with $m$ processors (or cores, \dots), one would
typically compute the $D \times m$ matrices $\mat{L}_s=\mUt\mM^s$
for $s=0,\dots,2\lceil D/m \rceil$ by having each processor compute a
sequence $\row{u}_i \mM^s$, where $\row{u}_i$ is the $i$th row of
$\mUt$. From these values, we then deduce the matrices
$\seqelt{s,\mU,\mV}=\mUt \mM^s \mV = \mat{L}_s\mV$ for
$s=0,\dots,2\lceil D/m \rceil$. Note that in our description, we
assume for simplicity that memory is not an issue, so that we can
store all needed elements from e.g. sequence $\mat{L}_s$; we discuss
this in more detail in \cref{ssec:mainalgo}.

Our main algorithm will also resort to scalar numerators of
the form $\Omega((\row{u}_i \mM^s \col{w})_{s \ge 0},
\minpoly)$, where $\row{w}$ is a given vector in $\K^{D \times 1}$ and
$\minpoly$ is the minimal polynomial of $\mM$. Since $\minpoly$ may
have degree $D$, and then $\Omega((\row{u}_i \mM^s \col{w})_{s \ge 0},
\minpoly)$ itself may have degree $D-1$,
the definition of $\Omega$ suggests that we
may need to compute up to $D$ terms of the sequence $\row{u}_i \mM^s
\col{w}$, which we would of course like to avoid. We now present an
alternative solution which involves solving a univariate polynomial linear system
and computing a matrix numerator, but only uses the sequence elements
$\mat{L}_s= \mUt \mM^s$ for $s=0,\dots,\lceil D/m \rceil-1$
which have already been computed.

Fix $i$ in $1,\dots,m$ and let $\row{a}_i$ be the row vector defined
by $$\row{a}_i =[0~\cdots~0~\minpoly~0~\cdots~0]  (\mat{P}_{\mU,\mV})^{-1} ,$$
where the minimal polynomial $\minpoly$ appears at the $i$th entry  of the
left-hand row vector. 
\begin{lemma}\label{utilde}
  For generic $\mU,\mV$, the row vector $\row{a}_i$ has entries which are
  polynomials of degree at most~$\deg(P) \le D$.
\end{lemma}
\begin{proof}
  For generic $\mU,\mV$, we saw that $\minpoly$ is the largest invariant factor
  of $ \mat{P}_{\mU,\mV}$; thus, the product $\minpoly\
  (\mat{P}_{\mU,\mV})^{-1}$ has polynomial entries. Since $\row{a}_i$ the $i$th
  row of this matrix, $\row{a}_i$ has polynomial entries.  Now, since
  $\mat{P}_{\mU,\mV}$ is reduced, the predictable degree property
  \citep[Theorem~6.3-13]{Kailath80} holds; it implies that each entry of
  $\row{a}_i$ has degree at most the maximum of the degrees of the entries of
  $\row{a}_i\mat{P}_{\mU,\mV}$. This maximum is $\deg(\minpoly)$.
\end{proof}
To compute $\row{a}_i$, we use again Storjohann's high-order lifting;
according to \cref{eqn:hol_cost}, the cost is $ O(m^{\omega-1} \M(D)
\log(D) \log(m)) \subset \softO{m^{\omega-1}D}$ operations in $\K$.
Once $\row{a}_i$ is known, the following lemma shows that we can
recover the scalar numerator $\Omega((\row{u}_i \mM^s \col{w})_{s \ge
0}, \minpoly)$ as a dot product.
\begin{lemma}\label{lemma:omegaOmega}
  For a generic choice of $\mU$ and $\mV$, and for any $\col{w}$ in
  $\K^{D \times 1}$, $ \mat{P}_{\mU,\mV}$ is a nonsingular matrix
  relation for the sequence $\mat{\mathcal{E}} =
  (\col{e}_s)_{s\ge0}$, with $\col{e}_s=\mUt \mM^s\, \col{w}$
  for all $s$, and we have
  \begin{align}\label{eq:OmegaOmega}
    [~\Omega((\row{u}_i \mM^s \col{w})_{s \ge 0}, \minpoly)~] = \row{a}_i\cdot \mat{\Omega}(\mat{\mathcal{E}}, \mat{P}_{\mU,\mV})
  \end{align}
  with $\row{a}_i$ in $\K[\var]^{1 \times m}$ and 
  $\mat{\Omega}(\mat{\mathcal{E}} , \mat{P}_{\mU,\mV}) \in \K[\var]^{m \times 1}$.
\end{lemma}
\begin{proof}
  The first item in \cref{randXY} shows that for a generic choice of
  $\mU$ and $\mV$, $\mat{P}_{\mU,\mV}$ cancels the sequence $(\mUt
  \mM^s)_{s \ge 0}$, and thus the sequence $\mat{\mathcal{E}}$ as well;
  this proves the first point. Then, the equality
  in \cref{eq:OmegaOmega} directly follows from the definitions:
  \begin{align*}
    [~ \Omega((\row{u}_i \mM^s \col{w})_{s \ge 0}, \minpoly)~]  &= [~\minpoly~]\ \sum_{s \ge 0} \frac{\row{u}_i \mM^s \col{w}}{T^{s+1}}\\
                                                                &=  [~0~\cdots~0~\minpoly~0~\cdots~0~]\  \sum_{s \ge 0} \frac{\mUt \mM^s \col{w}}{T^{s+1}}\\
                                                                &=  [~0~\cdots~0~\minpoly~0~\cdots~0~]\ (\mat{P}_{\mU,\mV})^{-1} \mat{P}_{\mU,\mV} \sum_{s \ge 0} \frac{\mUt \mM^s \col{w}}{T^{s+1}}\\
                                                                &=  \row{a}_i\cdot \mat{\Omega}(\mat{\mathcal{E}}, \mat{P}_{\mU,\mV}).
                                                                \qedhere
  \end{align*}
\end{proof}

The algorithm to obtain the scalar numerator $\Omega((\row{u}_i \mM^s
\col{w})_{s \ge 0}, \minpoly)$ follows; in the algorithm, we
assume that we know $\mat{P}_{\mU,\mV}$, $P$, $\row{a}_i$ and the
matrices $\mat{L}_s=\mUt \mM^s \in \K^{m \times D}$ for
$s=0,\dots,\lceil D/m \rceil-1$.

\begin{algorithm}[ht]
  \caption{{\sf ScalarNumerator}($\mat{P}_{\mU,\mV}, \minpoly, \row{w}, i, \row{a}_i,(\mat{L}_s)_{0 \le s < \lceil D/m\rceil}$)}
  {\bf Input:} \vspace{-0.5em}
  \begin{itemize}
    \item a minimal generator $\mat{P}_{\mU,\mV}$ of $(\mUt \mM^s \mV)_{s \ge 0}$
    \item the minimal polynomial $P$ of $\mM$
    \item $\row{w}$ in $\K^{D \times 1}$
    \item $i$ in $\{1,\dots,m\}$
    \item $\row{a}_i =  [0~\cdots~0~\minpoly~0~\cdots~0]  (\mat{P}_{\mU,\mV})^{-1} \in \K[T]^{1\times m}$
      where $\minpoly$ appears at the $i$th entry
    \item $\mat{L}_s = \mUt \mM^s \in \K^{m\times D}$, for $s=0,\dots,\lceil D/m\rceil-1$
  \end{itemize}
  {\bf Output:}  \vspace{-0.5em}
  \begin{itemize}
    \item         the scalar numerator $\Omega((\row{u}_i \mM^s \col{w})_{s \ge 0}, \minpoly) \in \K[T]$
  \end{itemize}
  \begin{enumerate}
    \item compute $\col{E}_s = \mat{L}_s \col{w} \in \K^{m \times 1}$ for $s=0,\dots,\lceil D/m\rceil-1$
    \item use these values to compute the matrix numerator $ \mat{\Omega}(\mat{\mathcal{E}}, \mat{P}_{\mU,\mV}) \in \K[T]^{m \times 1}$ 
      by \cref{eq:computeOmega}
    \item {\bf return} the entry of the $1 \times 1$ matrix $\row{a}_i\cdot  \mat{\Omega}(\mat{\mathcal{E}}, \mat{P}_{\mU,\mV})$ 
  \end{enumerate}
  \label{algo:scalar_numerator}
\end{algorithm}

Computing the first $\lceil D/m \rceil$ values of the sequence
$\mat{\mathcal{E}}=(\col{e}_s)_{s \ge 0}$ is done by using the
equality $\col{E}_s = \mat{L}_s \col{w}$ and takes $O(D^2)$ field
operations. Then, applying the cost estimate given after
\cref{eq:computeOmega}, we see that we have enough terms to compute
$\mat{\Omega}(\mat{\mathcal{E}}, \mat{P}_{\mU,\mV}) \in \K[T]^{m
\times 1}$ and that it takes $O(m^2 \M(D/m)) \subset O(m \M(D))$ operations in
$\K$. Then, the dot product with $\row{a}_i$ takes $O(m \M(D))$
operations, since both vectors have size $m$ and entries of degree at
most $D$. Thus, the runtime is 
\(O(D^2 + m \M(D)) \subset \softO{D^2}\)
operations in $\K$.

\section{Sequences associated to a zero-dimensional ideal}\label{sec:seq0}

We now focus on our main question: computing a zero-dimensional
parametrization of an algebraic set of the form $V=V(I)$, for some
zero-dimensional ideal $I$ in $\K[X_1,\dots,X_n]$. We write
$V=\{\balpha_1,\dots,\balpha_\dg\},$ with
$\balpha_i=(\alpha_{i,1},\dots,\alpha_{i,n}) \in \Kbar{}^n$ for all
$i$.  We also let $\D$ be the dimension of
$\residueI=\K[X_1,\dots,X_n]/I$, so that $\dg \le \D$, and {\em we
  assume that ${\rm char}(\K)$ is greater than $D$}.

In this section, we recall and expand on results from the appendix
of~\citep{BoSaSc03}, with the objective of computing a zero-dimensional
parametrization of $V$. These results were themselves inspired by 
those in~\citep{Rouillier99}, the latter being devoted to computations
with the trace linear form $\mathrm{tr}: \residueI\to\K$.
At this stage, we do not discuss data structures or complexity (this
is the subject of the next section); the algorithm in this
section simply describes what polynomials should be computed in order
to obtain a zero-dimensional parametrization.

\subsection{The structure of the dual}\label{ssec:dual}

For $i$ in $\{1,\dots,\dg\}$, let $\residueI_i$ be the local algebra at
$\balpha_i$, that is $\residueI_i=\Kbar[X_1,\dots,X_n]/I_i$, with $I_i$ the
$\m_{\balpha_i}$-primary component of $I$. By the Chinese Remainder
Theorem, $\residueI\otimes_\K \Kbar=\Kbar[X_1,\dots,X_n]/I$ is isomorphic to
the direct product $\residueI_1\times \cdots \times \residueI_\dg$.  We let $N_i$ be the
{\em nil-index} of $\residueI_i$, that is, the maximal integer $N$ such that
$\m_{\alpha_i}^N$ is not contained in $I_i$; for instance, $N_i=0$ if
and only if $\residueI_i$ is a field, if and only if $\balpha_i$ is a
nonsingular root of $I$. We also let $\D_i=\dim_\Kbar(\residueI_i)$, so that
we have $D_i \ge N_i$ and $\D=\D_1 + \cdots + \D_\dg$.

The sequences we consider below are of the form $(\ell(\lf^s))_{s \ge
0}$, for $\ell$ a $\K$-linear form $\residueI \to \K$ and $\lf$ in $\residueI$
(we will often write $\ell \in {\rm hom}_\K(\residueI,\K)$). For
such sequences, the following standard result will be useful
(see e.g.~\citep[Propositions~1 \& 2]{BoSaSc03} for a proof).
\begin{lemma}\label{lemma:minpoly}
  Let $\lf$ be in $\residueI$ and let $P \in \K[T]$ be its minimal
  polynomial. For a generic choice of $\ell$ in ${\rm hom}_\K(\residueI,\K)$,
  $P$ is the minimal polynomial of the sequence $(\ell(\lf^s))_{s \ge
  0}$.
\end{lemma}

The following results are classical; they go back
to~\citep{Macaulay16}, and have been used in computational algebra
since the 1990's~\citep{MaMoMo96,Mourrain97}. Fix $i$ in $1,\dots,\dg$.
There exists a basis of the dual ${\rm hom}_\Kbar(\residueI_i,\Kbar)$
consisting of linear forms $(\lambda_{i,j})_{1\le j \le \D_i}$ of the
form
$$\lambda_{i,j}: f \mapsto (\Lambda_{i,j}(f))(\balpha_i),$$
where $\Lambda_{i,j}$ is the operator
$$f \mapsto \Lambda_{i,j}(f) = \sum_{\mu=(\mu_1,\dots,\mu_n) \in
S_{i,j}} c_{i,j,\mu} \frac{ \partial^{\mu_1 + \cdots + \mu_n} f}
{\partial X_1^{\mu_1} \cdots \partial X_n^{\mu_n}},$$ for some finite
subset $S_{i,j}$ of $\N^n$ and nonzero constants $c_{i,j,\mu}$ in
$\Kbar$. 
For instance, when $\balpha_i$ is nonsingular, we have $D_i=1$, so
there is only one function $\lambda_{i,j}$, namely $\lambda_{i,1}$; we
write it $\lambda_{i,1}(f) = f(\balpha_i)$.
More generally, we can always take $\lambda_{i,1}$ of the form
$\lambda_{i,1}(f) = f(\balpha_i)$; for $j>1$, we can then also assume
that $S_{i,j}$ does not contain $\mu=(0,\dots,0)$ (that is, all terms
in $\Lambda_{i,j}$ have order $1$ or more). Thus, introducing new
variables $(U_{i,j})_{j =1,\dots,D_i}$, we deduce the existence of
nonzero homogeneous linear forms $P_{i,\mu}$ in
$(U_{i,j})_{j=1,\dots,D_i}$ such that for any $\ell$ in ${\rm
hom}_\Kbar(\residueI_i,\Kbar)$, there exists $\bu_i=(u_{i,j}) \in
\Kbar{}^{D_i}$ such that we have
\begin{align}\label{ell_param}
  \ell: f \mapsto \ell(f)
&= \sum_{j=1}^{D_i} u_{i,j} \lambda_{i,j}(f)\nonumber\\
&= \sum_{j=1}^{D_i} u_{i,j} \big(\Lambda_{i,j}(f)\big)(\balpha_i)\nonumber\\
&= \sum_{j=1}^{D_i} u_{i,j}
\sum_{\mu=(\mu_1,\dots,\mu_n) \in
S_{i,j}} c_{i,j,\mu} \frac{ \partial^{\mu_1 + \cdots + \mu_n} f}
{\partial X_1^{\mu_1} \cdots \partial X_n^{\mu_n}}(\balpha_i)\nonumber\\
&= \sum_{\mu=(\mu_1,\dots,\mu_n) \in S_i} P_{i,\mu}(\bu_i)
\frac{ \partial^{\mu_1 + \cdots + \mu_n} f}
{\partial X_1^{\mu_1} \cdots \partial X_n^{\mu_n}}(\balpha_i),
\end{align}
where $S_i$ is  the union of $S_{i,1},\dots,S_{i,D_i}$,
with in particular $P_{i,(0,\dots,0)}=U_{i,1}$ and where $P_{i,\mu}$
depends only on $(U_{i,j})_{j =2,\dots,D_i}$ for all $\mu$ in $S_i$,
$\mu \ne (0,\dots,0)$. Explicitly, we can write $P_{i,\mu}=\sum_{j\in
\{1,\dots,D_i\} \;\mid\; \mu \in S_{i,j}} c_{i,j,\mu}
U_{i,j}$. 

Fix $\ell$ nonzero in ${\rm hom}_\Kbar(\residueI_i,\Kbar)$, written
as in \cref{ell_param}. We can then
define its {\em order} $w$ and {\em symbol} $\pi$. The former is
the maximum of all $|\mu|=\mu_1+\cdots+\mu_n$ for
$\mu=(\mu_1,\dots,\mu_n)$ in $S_i$ such that $P_{i,\mu}(\bu_i)$ is
nonzero; by~\citep[Lemma~3.3]{Mourrain97} we have $w \le
N_i-1$. Then, we let
$$\pi =\sum_{\mu \in S_i,\ |\mu|=w}P_{i,\mu}(\bu_i) X_1^{\mu_1} \cdots
X_n^{\mu_n}$$ be the {\em symbol} of $\ell$; by construction,
this is a nonzero polynomial. 

Finally, we say a word about global objects.  Fix a linear form $\ell:
\residueI \to \K$. By the Chinese Remainder Theorem, there exist unique
$\ell_1,\dots,\ell_\dg$, with $\ell_i$ in ${\rm hom}_\Kbar(\residueI_i,\Kbar)$
for all $i$, such that the extension $\ell_\Kbar: \residueI\otimes_\K \Kbar
\to \Kbar$ decomposes as $\ell_\Kbar = \ell_1 + \cdots + \ell_\dg$.
Note that formally, we should write 
$\ell_\Kbar = \ell_1 \circ \phi_1 + \cdots + \ell_\dg \circ \phi_\dg$,
where for all $i$, $\phi_i$ is the canonical projection $\residueI \to \residueI_i$;
we will however omit these projection operators for simplicity.

We call {\em support} of $\ell$ the subset $\mathfrak{S}$ of
$\{1,\dots,\dg\}$ such that $\ell_i$ is nonzero exactly for $i$ in
$\mathfrak{S}$.  As a consequence, for all $f$ in $\residueI$, we have
\begin{equation}\label{eq:fui}
  \ell(f) = \ell_1(f) + \cdots + \ell_\dg(f)
  =  \sum_{i \in \mathfrak{S}} \ell_i(f).
\end{equation}
For $i$ in $\mathfrak{S}$, we denote by $w_i$ and $\pi_i$ respectively
the order and the symbol of $\ell_i$. For such a subset $\mathfrak{S}$
of $\{1,\dots,\dg\}$, we also write $\residueI_\mathfrak{S}=\prod_{i
  \in \mathfrak{S}} \residueI_i$ and $V_\mathfrak{S}= \{\balpha_i \mid
i \in \mathfrak{S}\}$.

\subsection{A fundamental formula}  \label{ssec:genseries}

Let $\lf$ be in $\residueI$ and $\ell$ in ${\rm hom}_\K(\residueI,\K)$.  The
sequences $(\ell(\lf^s))_{s\ge 0}$, and more generally the sequences $(\ell(v
\lf^s))_{s\ge 0}$ for $v$ in $\residueI$, are the core ingredients of our
algorithm.  This is justified by the following lemma, which gives a
description of generating series of the form $\sum_{s \ge 0} \ell(v
\lf^s)/T^{s+1}$. A slightly less precise version of it is in~\citep{BoSaSc03};
the more explicit expression given here will be needed in the last section of
this paper.

\begin{lemma}\label{lemma:formula}
  Let $\ell$ be in ${\rm hom}_\K(\residueI,\K)$, with support $\mathfrak{S}$,
  and let $\{\pi_i \mid i \in \mathfrak{S}\}$ and $\{w_i \mid i \in
  \mathfrak{S}\}$ be the symbols and orders of $\{\ell_i \mid i \in \mathfrak{S}\}$,
  for $\{\ell_i \mid i \in \mathfrak{S}\}$ as in \cref{ssec:dual}.

  Let $\lf=t_1 X_1 + \cdots +t_n X_n$, for some $t_1,\dots,t_n$ in $\K$
  and let $v$ be in $\K[X_1,\dots,X_n]$. Then, we have the equality
  \begin{align}\label{eq:sumgenseries}
    \sum_{s \ge 0} \frac{\ell(v \lf^s)}{T^{s+1}} = \sum_{i \in \mathfrak{S}}
    \frac{ v(\balpha_i)\, w_i!\, \pi_{i}(t_1,\dots,t_n) +
    (T-\lf(\balpha_i))A_{v,i}} {(T-\lf(\balpha_i))^{w_{i}+1}},
  \end{align}
  for some polynomials $\{A_{v,i} \in \Kbar[T] \mid i \in \mathfrak{S}\}$ which
  depend on the choice of $v$ and are such that $A_{v,i}$ has degree less than
  $w_i$ for all $i$ in $\mathfrak{S}$.
\end{lemma}
\begin{proof}
  Take $v$ and $\lf$ as above. Consider first an operator of the form $f
  \mapsto \frac{ \partial^{|\mu|} f} {\partial X_1^{\mu_1} \cdots
  \partial X_n^{\mu_n}}$, where we write
  $|\mu|=\mu_1+\cdots+\mu_n$. Then, we have the following generating
  series identities, with coefficients in $\K(X_1,\dots,X_n)$:
  \begin{align*}
    \sum_{s \ge 0} 
    \frac{ \partial^{|\mu|} ( v \lf^s )} {\partial X_1^{\mu_1} \cdots
    \partial X_n^{\mu_n}}
    \frac{1}{T^{s+1}} 
    &=  \sum_{s \ge 0} 
    \frac{ \partial^{|\mu|} (v \lf^s/T^{s+1})} {\partial X_1^{\mu_1} \cdots
    \partial X_n^{\mu_n}}\\
    &=  
    \frac{ \partial^{|\mu|} } {\partial X_1^{\mu_1} \cdots
    \partial X_n^{\mu_n}}
    \left (\sum_{s \ge 0} \frac{v \lf^s}{T^{s+1}}\right ) \\
    &= \frac{ \partial^{|\mu|} } {\partial X_1^{\mu_1} \cdots
    \partial X_n^{\mu_n}}
    \left (\frac v{T-\lf} \right ) \\
    &= \left (v\, |\mu|!\,    \frac {1}{(T-\lf)^{|\mu|+1}} \prod_{1 \le k \le n} 
      \left (\frac{ \partial \lf} {\partial X_k} \right)^{\mu_k}
    \right ) + \frac{P_{|\mu|}}{(T-\lf)^{|\mu|}} + \cdots + \frac{P_{1}}{(T-\lf)}\\
    &=\left ( v\, |\mu|!\,    \frac {1}{(T-\lf)^{|\mu|+1}} \prod_{1 \le k \le n} 
      t_k^{\mu_k}
    \right ) + \frac{H}{(T-\lf)^{|\mu|}},
  \end{align*}
  for some polynomials $P_1,\dots,P_{|\mu|},H$ in $\K[X_1,\dots,X_n,T]$ that
  depend on the choices of $\mu$, $v$ and $X$, with $\deg_T(P_i) < i$
  for all $i$ and thus $\deg_T(H) < |\mu|$.

  Take now a $\Kbar$-linear combination of such operators, such as $f \mapsto
  \sum_{\mu \in R} c_\mu \frac{ \partial^{|\mu|} f } {\partial
  X_1^{\mu_1} \cdots \partial X_n^{\mu_n}}$ for some finite subset $R$ of
  $\N^n$. The corresponding generating series becomes
  \[
    \sum_{s \ge 0} \sum_{\mu \in R} c_\mu \frac{ \partial^{|\mu|} ( v
    \lf^s )} {\partial X_1^{\mu_1} \cdots \partial X_n^{\mu_n}}
    \frac{1}{T^{s+1}} = v\,\sum_{\mu \in R} \left( c_\mu |\mu|!\,  \frac {1}{(T-\lf )^{|\mu|+1}} \prod_{1
      \le k \le n} t_k^{\mu_k} \right )+\sum_{\mu
    \in R} \frac{H_\mu}{(T-\lf)^{|\mu|}},
  \]
  where each $H_\mu \in \Kbar[X_1,\dots,X_n,T]$ has degree in $T$ less than
  $|\mu|$.  Let $w$ be the maximum of all $|\mu|$ for $\mu$ in
  $R$. We can rewrite the above as
  \begin{align*}
    v\, w! 
    \sum_{\mu \in R, |\mu|=w}\left ( c_\mu
      \,    \frac {1}{(T-\lf )^{w+1}} \prod_{1 \le k \le n} 
    t_k^{\mu_k}\right )
    + \frac{A}{(T-\lf )^{w}},
  \end{align*}
  for some polynomial $A \in \Kbar[X_1,\dots,X_n,T]$ of degree less than $w$ in $T$. 
  Then, if we let 
  $\pi =\sum_{\mu \in R,\ |\mu|=w} c_{\mu} X_1^{\mu_1} \cdots
  X_n^{\mu_n}$, this becomes
  \begin{align*}
    \sum_{s \ge 0} 
    \sum_{\mu \in R} c_\mu \frac{ \partial^{|\mu|} ( v \lf^s )} { X_1^{\mu_1} \cdots
    X_n^{\mu_n}}
    \frac{1}{T^{s+1}} 
    &=
    v\, w! \,  \pi(t_1,\dots,t_n)
    \frac {1}{(T-\lf )^{w+1}}
    + \frac{A}{(T-\lf )^{w}}.
  \end{align*}
  Applying this formula to the sum $\ell=\sum_{i \in
  \mathfrak{S}}\ell_i$ from \cref{eq:fui}, and taking into account
  the expression in \cref{ell_param} for each $\ell_i$, we obtain the
  claim in the lemma.
\end{proof}

\subsection{Computing a zero-dimensional parametrization}  \label{ssec:abstractlago}

As a corollary, we recover the following result, that shows how to compute a
zero-dimensional parametrization of $V_{\mathfrak{S}}$ (see \cref{algo:para2}).
Our main usage of it will be with $\mathfrak{S}=\{1,\dots,\dg\}$, in which case
$V_{\mathfrak{S}}=V$, but in the last section of the paper, we will also work
with strict subsets. 

\begin{algorithm}[ht]
  \caption{$\mathsf{Parametrization}(\ell,\lf)$}  ~\\
  {\bf Input:} \vspace{-0.5em}
  \begin{itemize}\setlength\itemsep{0em}
    \item  a linear form $\ell$ over $\residueI_\mathfrak{S}$
    \item $\lf=t_1 X_1 + \cdots + t_n X_n$
  \end{itemize}
  {\bf Output:}  \vspace{-0.5em}
  \begin{itemize}
    \item              polynomials $((\sqfree,V_1,\dots,V_n),\lf)$, with $\sqfree,V_1,\dots,V_n$ in $\K[T]$
  \end{itemize}
  \begin{enumerate}\setlength\itemsep{0em}
    \item let $\minpoly$ be the minimal polynomial of the sequence $(\ell(\lf^s))_{s \ge 0}$
    \item let $\sqfree$ be the squarefree part of $\minpoly$
    \item let $C_1 = \Omega((\ell(\lf^s))_{s\ge0} ,\minpoly)$
    \item \textbf{for} $i=1,\dots,n$ \textbf{do} \\
      \phantom{for} let $C_{X_i} = \Omega((\ell(X_i \lf^s))_{s\ge0}, \minpoly)$ 
    \item \textbf{return} $((\sqfree, C_{X_1}/ C_1 \bmod \sqfree, \dots, C_{X_n}/ C_{1} \bmod \sqfree),\lf)$
  \end{enumerate}
  \label{algo:para2}
\end{algorithm}

\begin{lemma}[{\citep[Proposition~3]{BoSaSc03}}]
  \label{lemma:para2}
  Let $\ell$ be a generic element of ${\rm
  hom}_{\Kbar}(\residueI_\mathfrak{S},\Kbar)$ and 
 $\lf = t_1 X_1 + \cdots + t_n X_n$
 be a generic 
  linear combination of $X_1,\dots,X_n$. Then the output $((\sqfree,V_1,\dots,V_n),\lf)$ of
  $\mathsf{Parametrization}(\ell,\lf)$ is a zero-dimensional
  parametrization of $V_{\mathfrak{S}}$.
\end{lemma}

The proof in~\cite{BoSaSc03} is written for
$\mathfrak{S}=\{1,\dots,\kappa\}$, but works equally as well for the
more general case; it uses a weaker form of the previous lemma, that
is sufficient in this case.  We demonstrate how this algorithm works
through a small example, in which we already know the coordinates of
the solutions. Let
$$I = \langle (X_1-1)(X_2-2),(X_1-3)(X_2-4)\rangle \subset
\F_{101}[X_1,X_2].$$ Then, $V(I) = \{\balpha_1,\balpha_2\}$,
with $\balpha_1= (1,4)$ and $\balpha_2=(3,2)$; we take
$\lf=X_1$, which separates the points of $V(I)$.  We choose the linear form
\[
  \ell: \F_{101}[X_1,X_2]/I \to \F_{101},\;
  f \mapsto \ell(f) = 17 f(\balpha_1) + 33 f(\balpha_2);
\]
then, $\mathfrak{S}=\{1,2\}$, and the symbols $\pi_1$ and $\pi_2$ 
are respectively the constants 17 and 33. We have
\begin{align*}
  \ell(X_1^s) &= 17 \cdot 1^s + 33 \cdot 3^s,\\
  \ell(X_2X_1^s) &= 17 \cdot 4 \cdot 1^s + 33 \cdot 2 \cdot 3^s.
\end{align*} 
We associate a generating series to each sequence:
\begin{align*}
  Z_1 = \sum_{s \ge 0} \frac{\ell(X^s_1)}{T^{s+1}}
&= \frac{17}{T-1} + \frac{33}{T-3}
= \frac{17(T-3)+33(T-1)}{(T-1)(T-3)}, \\
Z_{X_2} = \sum_{s\ge0} \frac{\ell(X_2X_1^s)}{T^{s+1} }
&= \frac{17\cdot 4}{T-1} + \frac{33 \cdot 2}{T-3}
= \frac{17\cdot 4 (T-3) + 33\cdot 2(T-1)}{(T-1)(T-3)}.
\end{align*}
These generating series have for common denominator $\minpoly = (T-1)(T-3)$,
whose roots are the coordinates of $X_1$ in $V(I)$;
their numerators are respectively
\begin{align*}
  C_{1} &= \Omega((\ell(X^s_1))_{s\ge 0},\minpoly) = 17 (T-3) + 33(T-1), \\
  C_{X_2} &= \Omega((\ell(X_2X^s_1))_{s\ge 0},\minpoly) = 17\cdot 4 (T-3) + 33\cdot 2(T-1).
\end{align*}
Now, let
\[
  V_2 
  =\frac{C_{X_2}}{C_1} \bmod \minpoly
  =\frac{17\cdot 4 (T-3) + 33\cdot 2(T-1)}{17(T-3)+33(T-1)} \bmod \minpoly
  =100 T +5.
\]
Then, $V_2(1) = 4$ and $V_2(3) = 2$, as expected.

\section{The main algorithm}\label{sec:main}

In this section, we extend the algorithm of~\citep{BoSaSc03} to compute
a zero-dimensional parametrization of $V(I)$, for some
zero-dimensional ideal $I$ of $\K[X_1,\dots,X_n]$, by using blocking
methods. Our input is a monomial basis
$\basis=(b_1,\dots,b_D)$ of $\residueI=\K[X_1,\dots,X_n]/I$, together
with the multiplication matrices $\mM_1,\dots,\mM_n$ of respectively
$X_1,\dots,X_n$ in this basis; for definiteness, we suppose that the
first basis element in $\basis$ is~$b_1=1$. As above, $D$ denotes the
dimension of $\residueI$.

The first subsection presents the main algorithm. Its main feature is
that after we compute the Krylov sequence used to find a minimal
matrix generator, we recover all entries of the output for a minor
cost, without computing another Krylov sequence. We make no assumption
on $I$ (radicality, shape position, \dots), except of course that it
has dimension zero; however, we assume (as in the previous subsection)
that the characteristic of $\K$ is greater than $D$. 

Then, in the second subsection we present a simple example, and in the third we
show experimental results of an implementation based on the C++ libraries
Eigen, LinBox and NTL.

\subsection{Description, correctness and cost analysis}\label{ssec:mainalgo}

We mentioned that the method of \citet{Steel15} already uses the block
Wiedemann algorithm to compute the minimal polynomial $\minpoly$ of
$\lf=t_1 X_1 + \cdots + t_n X_n$; {given sufficiently many terms of the
  sequence $(\mUt \mM^s \mV)$, this is done by means of polynomial
  lattice reduction}. Knowing the roots of $\minpoly$ in $\K$, that
algorithm uses an ``evaluation'' method for the rest (several
Gr\"obner basis computations, all with one variable less).

Our algorithm (\cref{algo:block-sparse-fglm}) computes the whole
zero-dimensional parametrization of $V(I)$ for essentially the same cost as the
computation of the minimal polynomial. Hereafter, $\col{\varepsilon}_1$
denotes the size-$D$ column vector whose only nonzero entry is a $1$ at the
first index: $\col{\varepsilon}_1 =\trsp{[1~0~\cdots~0]}$.

\begin{algorithm}[ht]
  \caption{$\mainalgoname(\mM_1,\dots,\mM_n,\mU,\mV,\lf)$}
  {\bf Input:} \vspace{-0.5em}
  \begin{itemize}
    \item $\mM_1,\dots,\mM_n$ multiplication matrices defined as above
    \item  $\mU,\mV \in \mathbb{K}^{D \times m}$, for some block dimension  $m \in \{1,\dots,D\}$
    \item $\lf =t_1 X_1 + \cdots + t_n X_n$
  \end{itemize}
  {\bf Output:}  \vspace{-0.5em}
  \begin{itemize}
    \item         polynomials $((\sqfree,V_1,\dots,V_n),\lf)$, with $\sqfree,V_1,\dots,V_n$ in $\K[T]$
  \end{itemize}
  \begin{enumerate}
    \item\label{mainstep1}   let $\mM = t_1 \mM_1 + \cdots + t_n \mM_n$
    \item\label{mainstep3} { compute $\mat{L}_s = \mUt\mM^s$ for $s=0,\dots,2d-1$, with $d = \lceil D/m \rceil$}
    \item\label{mainstep4} { compute $\seqelt{s,\mU,\mV}= \mat{L}_s\mV$ for $s=0,\dots, 2d-1$}
    \item\label{mainstep5} { compute a minimal matrix generator $\mat{P}_{\mU,\mV}$ of $(\seqelt{s,\mU,\mV})_{0 \le s < 2d}$}
    \item\label{mainstep6} { let $\minpoly$ be the largest invariant factor of $\mat{P}_{\mU,\mV}$}
    \item\label{mainstep7} { let $\sqfree$ be  the squarefree part  of $\minpoly$}
    \item\label{mainstep8} { let $\row{a}_1 = [P~0 ~\cdots~ 0] (\mat{P}_{\mU,\mV})^{-1}$}
    \item\label{mainstep9}  let $C_1 = \mathsf{{\sf ScalarNumerator}}(\mat{P}_{\mU,\mV}, \minpoly, \col{\varepsilon}_1, 1, \row{a}_1, 
      (\mat{L}_s)_{0 \le s < d})$
    \item\label{mainstep10} \textbf{for} $i=1,\dots,n$ \textbf{do} \\
      \phantom{for}  let $C_{X_i} = \mathsf{{\sf ScalarNumerator}}(\mat{P}_{\mU,\mV}, \minpoly, \mM_i\col{\varepsilon}_1, 1, \row{a}_1, (\mat{L}_s)_{0 \le s < d})$
    \item\label{mainstep11}     \textbf{return} $((\sqfree, C_{X_1}/ C_1 \bmod \sqfree, \dots, C_{X_n}/ C_{1} \bmod \sqfree),\lf)$
  \end{enumerate}  \label{algo:block-sparse-fglm}
\end{algorithm}

We first prove correctness of the algorithm, for generic choices of
$t_1,\dots,t_n$, $\mU$ and $\mV$. The first step computes the
multiplication matrix $\mM=t_1 \mM_1 + \cdots + t_n \mM_n$ of $\lf=t_1
X_1 + \cdots + t_n X_n$.
Then, we compute the first $2d$ terms of the sequence $\seq_{\mU,\mV}=
(\mUt\mM^s\mV)_{s\ge 0}$. As discussed in \cref{ssec:appliW}, \cref{randXY} shows that
the matrix polynomial $\mat{P}_{\mU,\mV}$ is indeed a minimal left
generator of the sequence $\seq_{\mU,\mV}$, that $\minpoly$ is the minimal 
polynomial of $\lf$ and $\sqfree$ its squarefree part.

We find the rest of the polynomials in the output by following
\cref{algo:para2}. In particular, the scalar numerators
needed in this algorithm are computed using \cref{algo:scalar_numerator}
($\mathsf{ScalarNumerator}$); indeed, applying \cref{lemma:omegaOmega},
we see that calling this algorithm
at Steps~\ref{mainstep9} and~\ref{mainstep10}
computes $$C_1 = \Omega((\row{u}_i \mM^s \col{\varepsilon}_1)_{s \ge
0}, \minpoly) \quad\text{and}\quad C_{X_i} = \Omega((\row{u}_1 \mM^s
\mM_i \col{\varepsilon}_1)_{s \ge 0}, \minpoly),\ \ i=1,\dots,n.$$ 

Let $\ell:\residueI \to \K$ be the linear form $f = \sum_{i=1}^D f_i b_i
\mapsto \sum_{1 \le i \le D} f_i u_{i,1}$, where $u_{i,1}$ is the entry at
position $(i,1)$ in $\mU$. The two polynomials above can be rewritten as 
\[
  C_1 = \Omega( (\ell(X^s))_{s \ge 0}, \minpoly) \quad\text{and}\quad
C_{X_i} = \Omega( (\ell(X_i X^s))_{s \ge 0}, \minpoly),
\]
so they coincide with the polynomials in \cref{algo:para2}.  Thus, by
\cref{lemma:para2}, for generic $\mU$ and $\lf$ the output of $\mainalgoname$
is indeed a zero-dimensional parametrization of $V(I)$.  

\begin{remark}
  As already pointed out in \cref{ssec:scalar_numer}, the algorithm is
  written assuming that memory usage is not a limiting factor (this
  makes it slightly easier to write the pseudo-code). As described
  here, the algorithm stores $\Theta(D^2)$ field elements in the
  sequence $\mat{L}_s$ computed at Step~\ref{mainstep3}, since they
  are re-used at Steps~\ref{mainstep9} and~\ref{mainstep10}.  We may
  instead discard each matrix $\mat{L}_s$ after it is used, by
  computing on the fly the column vectors needed for Steps~\ref{mainstep9}
  and~\ref{mainstep10}.

  If the multiplication matrices are dense, little is gained this way (in
  the worst case, they use themselves $n D^2$ field elements), but savings can
  be substantial if these  matrices are sparse.
\end{remark}

For the cost analysis, we focus on a sparse model: we let $\density
\in [0,1]$ denote the density of $\mM$ and the $\mM_i$'s, that is, all
these matrices have at most $\density D^2$ nonzero entries.  As a
result, a matrix-vector product by $\mM$ can be done in $O(\density
D^2)$ operations in $\K$. In particular, the cost incurred at
Step~\ref{mainstep1} to compute $\mM$ is $O(\density n D^2)$.

In this context, the main purpose of Coppersmith's blocking strategy
is to allow for easy parallelization. Computing the matrices
$\mat{L}_s=\mUt\mM^s$, for $s=0,\dots,2d-1$, is the bottleneck
of the algorithm but can be parallelized. This is done by
working row-wise, computing independently the sequences
$(\row{\ell}_{i,s})_{0 \le s < 2d}$ of the $i$th rows of
$(\mat{L}_s)_{0 \le s < 2d}$ as $\row{\ell}_{i,0}=\row{u}_i$ and
$\row{\ell}_{i,s+1} = \row{\ell}_{i,s}
\mM$ for all $i,s$, where $\row{u}_i$ is the $i$th row of $\mUt$.
For a fixed $i \in \{1,\dots,m\}$, computing $(\mat{\ell}_{i,s})_{0
\le i < 2d}$ costs $O(d \density D^2) = O(\density D^3/m )$ field operations. If
we are able to compute in parallel $m$ vector-matrix products at once,
the {\em span} of Step~\ref{mainstep3} is thus $O(\density D^3/m)$, whereas
the total work is $O(\density D^3)$.

At Step~\ref{mainstep4}, we can then compute $\seqelt{s,\mU,\mV}=\mUt\mM^s\mV$, for $s=0,\dots,2d-1$ by the
product
$$
\begin{bmatrix}
  \mUt\\
  \mUt \mM\\
  \mUt \mM^2\\
  \vdots\\
  \mUt \mM^{2d-1}\\
\end{bmatrix} \mV
= 
\begin{bmatrix}
  \mUt \mV\\
  \mUt \mM \mV\\
  \mUt \mM^2 \mV\\
  \vdots \\
  \mUt \mM^{2d-1} \mV\\
\end{bmatrix}
$$
of size $O(D) \times D$ by $D \times m$; since $m \le D$, this  costs $O(m^{\omega-2}D^2)$
base field operations.

Recall from \cref{section:matrix_seq} that we can compute a minimal
matrix generator $\mat{P}_{\mU,\mV}$ in time
\[
  O(m^{\omega} \M(D/m) \log(D/m)) \subseteq O(m^{\omega-1} \M(D) \log(D)),
\]
and from
\cref{ssec:appliW,ssec:scalar_numer} that the
largest invariant factor $P$ and the vector $\row{a}_1$ can be
computed in time $O(m^{\omega-1} \M(D) \log(D) \log(m))$.  Computing $\sqfree$
takes time $O(\M(D) \log(D))$.

In \cref{ssec:scalar_numer}, we saw that each call to {\sf ScalarNumerator}
takes $O(D^2 + m\M(D))$ operations, for a total of $O(nD^2 + n
m\M(D))$; the final computations modulo $\sqfree$ at Step \ref{mainstep11} take time $O(n\M(D) + \M(D) \log(D))$.
Thus, altogether, assuming perfect parallelization 
at Step~\ref{mainstep3}, the total span is
$$O\left (\density \frac{D^3}m + m^{\omega-1} \M(D) \log(D) \log(m) + nD^2 + nm\M(D)\right ),$$
and the total work is
$$O\left (\density D^3 + m^{\omega-1} \M(D) \log(D) \log(m) + nD^2 +
nm\M(D)\right ).$$ Although one may work on parallelizing other steps than
Step~\ref{mainstep3}, we note that this step is simultaneously the most costly in
theory and in practice, and the easiest to parallelize.

\begin{remark}
  Our algorithm only computes the first invariant factor of
  $\mat{P}_{\mU,\mV}$, that is, of $T \mat{I}_D-\mM$. A natural
  question is whether computing further invariant factors can be of
  any use in the algorithm (or possibly can help us determine part of
  the structure of the algebras $\residueI_i$).
\end{remark}

We conclude this section by a discussion of ``dense'' versions of the
algorithm (to be used when the density $\density$ is close to $1$). 
If we use a dense model for our matrices, our algorithms
should rely on dense matrix multiplication. We will see two possible
approaches, which respectively take $m=1$ and $m=D$; we will not
discuss how they parallelize, merely pointing out that one may simply
parallelize dense matrix multiplications throughout the algorithms.

Let us first discuss the modifications in the algorithm to apply if we choose
$m=1$. We compute the
row-vectors $\mat{L}_s$, for $s=0,\dots,2D-1$, using the
square-and-multiply technique from \citep{Keller85},
for $O(D^\omega \log(D))$ operations in
$\K$. For generic choices of $\mU$ and $\mV$, a minimal matrix generator
$\mat{P}_{\mU,\mV}$ is equal to the minimum
polynomial $\minpoly$ of $\mM$, and can be computed efficiently by the
Berlekamp-Massey algorithm; besides,
$\row{a}_1 = P (\mat{P}_{\mU,\mV})^{-1} = 1$. Computing the scalar
numerators is simply a power series multiplication in degree at most
$D$. Altogether, the runtime is $O(D^{\omega} \log(D) + nD^2)$, where
the second term gives the cost of computing $\mM$.

When $m = D$, $\mU,\mV \in \mathbb{K}^{D \times D}$ are square
matrices and $d = D/m = 1$. The canonical matrix generator of
$\seq_{\mU,\mV} = (\mUt\mV, \mUt\mM\mV,\dots)$ is
$\mat{P}_{\mU,\mV} = T \mat{I}_D - \mUt\mM\itrsp{\mU}$ and its
largest invariant factor $\minpoly$ and $\row{a}_1$ can be computed in
$O(D^\omega \log(D))$ operations in $\K$ using high-order lifting.

The numerator $\mat{\Omega}(\mUt \mM \col{\varepsilon}_1,
\mat{P}_{\mU,\mV})$ is then seen to be $\mUt \col{\varepsilon}_1$, that is, the
first column of $\mUt$; we recover $C_1$ from it through a dot
product with $\row{a}_1$. Similarly, the numerator
$\mat{\Omega}(\mUt \mM \mM_i \col{\varepsilon}_1,
\mat{P}_{\mU,mV})$ is $\mUt \mM_i \col{\varepsilon}_1$, and gives
us $C_{X_i}$. Altogether, the runtime is again $O(D^{\omega} \log(D) +
nD^2)$, where the second term now gives the cost of computing $\mM$,
as well as $C_1$ and all $C_{X_i}$.

\subsection{Example}

We give an example of our algorithm with a non-radical system as input. Let
$$
I = 
\sbox0{$\begin{array}{c}
    X_1^3 + 88X_1^2 + 56X_1 + 21,\\
    X_1^2X_2 + 91X_1^2 + 92X_1X_2 + 90X_1 + 20X_2 + 2,\\
    X_1X_2^2 + 81X_1X_2 + 100X_1 + 96X_2^2 + 100X_2 + 5,\\
    X_1^2X_2 + 81X_1^2 + 93X_1X_2 + 59X_1 + 16X_2 + 84,\\
    X_1X_2^2 + 71X_1X_2 + 99X_1 + 97X_2^2 + 19X_2 + 8,\\
    X_2^3 + 61X_2^2 + 96X_2 + 20
\end{array}$}
\mathopen{\resizebox{1.2\width}{\ht0}{$\Bigg\langle$}}
  \usebox{0}
\mathclose{\resizebox{1.2\width}{\ht0}{$\Bigg\rangle$}}
\subset \F_{101}[X_1,X_2];
$$ the corresponding residue class ring $\residueI=\F_{101}[X_1,X_2]/I$
has dimension $D=4$.  Although it is not obvious from the generators,
the ideal $I$ is simply $\mathfrak{m}_1^2 \mathfrak{m}_2$, where
$\mathfrak{m}_1 =\langle X_1-4,X_2-10\rangle$ and $\mathfrak{m}_2
=\langle X_1-5,X_2-20\rangle$ (this immediately implies that $D=3+1=4$).

We choose $\lf = 2X_1 + 53 X_2$, so that the multiplication matrices of $X_1$, $X_2$ and $\lf$
in the basis $\basis=(1,X_1,X_2,X_2^2)$ of $\residueI$ are respectively
\begin{align*}
  \mM_1 = \begin{bmatrix}
    7   & 91 & 100& 0\\
    41  & 2  & 20 & 0\\
    100 & 10 & 8  & 1\\
    1   & 71 & 86 & 0
  \end{bmatrix},\,
  \mM_2 = \begin{bmatrix}
    40&  1 & 91 & 0\\
    5 &  0 &  2 & 1\\
    0 &  0 & 10 & 0\\
    81&  0 & 71 & 0
  \end{bmatrix},\,\text{and }
  \mM = \begin{bmatrix}
    13 & 33&  74&  0\\
    44 &  4&  45&  53\\
    99 & 20&  41&  2\\
    53 & 41&  97&  0
  \end{bmatrix}.
\end{align*}

We choose $m = 2$ and take $\mU,\mV \in \F_{101}^{D\times m}$ with 
entries
\[ \mU = \begin{bmatrix}
    84& 38\\
    29& 58\\
    80& 43\\
    7& 82
  \end{bmatrix},\quad
  \mV = \begin{bmatrix}
    6&  97\\
    83&  58\\
    0&  95\\
    59&  89
  \end{bmatrix}.
\]
We compute the first $2d=2\lceil D/m\rceil =4$ terms in the matrix
sequence $\seq_{\mU,\mV} = (\mUt\mM^s\mV)_{s\ge0}$ and its
minimum matrix generator $\mat{P}_{\mU,\mV}$. This is done
by first computing
\begin{align}\label{exampleproducts}
  \mUt =
  \begin{bmatrix}
    84&  29&  80&   7\\
    38&  58&  43&  82
  \end{bmatrix},
& 
\quad \mUt \mM 
=
\begin{bmatrix}
  54&  28&  67&  81\\
  34&  52&  90&  29
\end{bmatrix},
\\[2mm]
\mUt \mM^2 
=
\begin{bmatrix}
  33&  91&   3&  2\\
  47&  77&  47&  7
\end{bmatrix},
 &\quad \mUt \mM^3 
 =
 \begin{bmatrix}
   89&  80&  87&  82\\
   34&  56&  55&  34
 \end{bmatrix}, \nonumber
\end{align}
from which we get, by right-multiplication by $\mV$,
\[
  \seq_{\mU,\mV} =
  \begin{bmatrix}
    92& 75\\  
    83& 51
  \end{bmatrix},
  \begin{bmatrix}
    54& 34\\  
    70& 73
  \end{bmatrix},
  \begin{bmatrix}
    92& 54\\  
    16& 74
  \end{bmatrix},
  \begin{bmatrix}
    94& 51\\
    91& 51
  \end{bmatrix},\dots
\]
and then we obtain the minimal matrix generator
\[
  \mat{P}_{\mU,\mV} =
  \begin{bmatrix}
    T^2 + 60T + 62 &       88T + 25\\
    100T + 33 & T^2 + 84T + 78
  \end{bmatrix}.
\]
The largest invariant factor of $\mat{P}_{\mU,\mV}$ is 
$P = T^3 + 76T^2 + 100T + 7$,
with squarefree part $\sqfree=T^2+8T+61$.
Next, we compute the row vector $\row{a}_1 = [T+16,13]$ 
by solving $\row{a}_1 = [P~~0] (\mat{P}_{\mU,\mV})^{-1}$,
and the matrix numerator
\[
  \mat{\Omega}( (\mUt \mM^s\col{\varepsilon}_1)_{s \ge 0}, \mat{P}_{\mU,\mV}) 
  =\left [\begin{matrix} 84T + 55 \\ 38T + 11 \end{matrix}\right ]
\]
which is made from the entries of nonnegative degree in the product
\[\mat{P}_{\mU,\mV} 
  \left (
    \begin{bmatrix}
      84\\  
      38
    \end{bmatrix}\cdot \frac 1T+
    \begin{bmatrix}
      54\\  
      34
    \end{bmatrix}\cdot \frac 1{T^2}+
  \cdots \right),
\]
where the columns are the first columns of the
matrices in \cref{exampleproducts} (as per \cref{eq:computeOmega}, we only need $d=2$
terms in the right-hand side).  From this, we find the scalar
numerator 
\[
  C_1 = \Omega((\row{u}_1 \mM^s\col{\varepsilon}_1)_{s \ge
  0}, \minpoly) = 84T^2 + 75T + 13
\]
by means of the dot product $[~C_1~] = \row{a}_1\cdot
\mat{\Omega}( (\mUt \mM^s\col{\varepsilon}_1)_{s \ge 0},
\mat{P}_{\mU,\mV})$.

We then find $C_{X_1} = 88T^2 + 47T  + 16$, by proceeding similarly: we find
the matrix numerator
\[
  \mat{\Omega}( (\row{u}_1 \mM^s \mM_1\col{\varepsilon}_1)_{s \ge 0},
  \mat{P}_{\mU,\mV}) = \begin{bmatrix} 88T + 19 \\ 57T+40
  \end{bmatrix}
\]
and we take the dot product with $\row{a}_1$.
Thus, we obtain the polynomial $V_1 = C_{X_1}/C_1 \bmod \sqfree =
15T+14$. 
We compute $V_2= 49T+9$ in the same way,
and our output is 
$$((T^2+8T+61, 15T+14, 49T+9), 2X_1 + 53 X_2).$$ As a sanity check, we
recall that $V( I)$ has two points in $\F_{101}^2$, namely
$(4,10)$ and $(5,20)$; accordingly, $Q$ has two roots in $\F_{101}$,
$33$ and $60$, and we have $(V_1(33) = 4, V_2(33) = 10)$ and $(V_1(60) = 5,
V_2(60) = 20)$, as expected.

\subsection{Experimental results}\label{section:ex}

In \cref{tbl:timings_mainalgo}, we give the timings in seconds for
different values of $m$ for \cref{algo:block-sparse-fglm}.  Our
implementation is based on Shoup's NTL for univariate polynomials, LinBox
for dense polynomial matrices and matrix generator
computations, and Eigen for sparse matrix-vector
products. It is dedicated to small prime fields; thus,
as is done in~\citep{fflas-ffpack} for dense matrices, we use machine
floats to do the bulk of the calculations. Explicitly, we rely on
Eigen's {\tt SparseMatrix<double, RowMajor>} class to store our
multiplication matrices and do the matrix-vector products, reducing
the results modulo $p$ afterwards.

All timings are measured on an Intel Xeon CPU E5-2667 with 128GB RAM
and 8 cores (or 16 via hyperthreading). For each value of
$m$ in $\{1,3,6\}$, we create and run $m$ threads in parallel.

In all cases, we start from multiplication matrices computed from a
{degrevlex} Gr\"obner basis in Magma~\citep{BoCaPl97}. The timings
reported here do not include this precomputation; instead, we refer
the reader to~\citep{FaMo17} for extensive experiments comparing the
runtime of two similar algorithmic stages (degree basis computation and
conversion to a lex ordering).  Rather than optimizing our
implementation, our main focus here was to demonstrate the effects of
parallelization, and how the Krylov sequence computation dominates the
runtime in these examples. In particular, improvements are possible for
polynomial matrix computations (minimal matrix generator, largest invariant
factor, \dots), but this is by no means a bottleneck here.

All our inputs as well as our source code are available at
\url{https://github.com/vneiger/block-sparse-fglm}.
Some systems
are well-known (Katsura or Eco from~\citep{Morgan88}), while we also
consider families of randomly generated inputs (some are
inspired from~\citep{FaMo17}). Systems rand1-$i$ have $3$ variables and
randomly generated equations (of degree depending on $i$), and
rand2-$i$ are similar with $4$ variables. These systems are
generically radical and in shape position (for the projection on the
first coordinate axis; the last column uses that
convention). Systems mixed1-$i$ and mixed2-$i$ have similar numbers of
variables as the previous ones, but some of their solutions are
multiple, so the ideals are not radical (and not in shape
position). Systems mixed3-$i$ have only one multiple root (the others
are simple), in increasing numbers of variables; they are not in shape
position. Systems W1-$\kappa$-$n$-$p$ are determinantal equations that
describe the computation of critical points for the projection $
(\alpha_1,\dots,\alpha_n) \mapsto \alpha_1$ on $V(f_1,\dots,f_p)$,
where $f_1,\dots,f_p$ are $p$ equations of degree $\kappa$ in $n$
variables.

We used OpenMP {\tt parallel
for} pragmas to parallelize the computation of the Krylov sequence,
as described in \cref{ssec:mainalgo}.
In the columns $m=1$, $m=3$, $m=6$, the numbers in brackets indicate the
fraction of the total time spent in computing the Krylov sequence
$(\mUt \mM^s)_{0 \le s< 2d}$, with $d=\lceil D/m\rceil$. This is
always by far the dominant factor.  In other words, almost all the time is
spent doing sparse matrix-vector products for matrices with machine float
entries. 

Increasing $m$ has two effects. On the plus side, it
decreases the length of the Krylov sequence. On the other side, while
the algorithm performs better by a factor often close to 3 for $m=3$,
the gain is not as significant for $m=6$, so that parallelization is
less effective. It is an interesting question to clarify this issue.  Besides, increasing $m$ also increases the time to
compute the output polynomials, through minimal generator computation,
high-order lifting, \dots. However, from the observed speedup and the
ratios in the table, we can conclude that this effect is minor.

\begin{table}[ht]
  \centering
  \setlength\tabcolsep{4pt}
  \caption{Timings (in seconds) for polynomials over $\F_{65537}$}
  \label{tbl:timings_mainalgo}
  \begin{tabular}{c|c|c|c|c|c|c|c}
    \textbf{name}& $\bm{n}$ & $\bm{D}$ & \textbf{density} $\boldsymbol\rho$ & $\bm{m = 1}$ & $\bm{m = 3}$ & $\bm{m = 6}$ 
                 & \textbf{radical/shape} \\
                 \hline
    rand1-26&3 &17576&0.06& 692(0.98) & 307(0.969) & 168(0.926) & yes/yes \\
    rand1-28&3 &21952&0.05&1261(0.983) & 471(0.971) & 331(0.944)& yes/yes  \\
    rand1-30&3 &27000&0.05&2191(0.986) & 786(0.974) & 512(0.946)& yes/yes  \\
    rand2-10&4 &10000&0.14&301(0.981) & 109(0.964) & 79(0.934)& yes/yes  \\
    rand2-11&4 &14641&0.13&851(0.987) & 303(0.975) & 239(0.961) & yes/yes \\
    rand2-12&4&20736&0.12&2180(0.99) & 784(0.982) & 648(0.972)& yes/yes  \\
    mixed1-22&3 &10864&0.07&207(0.973) & 75(0.947) & 58(0.909) &no/no \\
    mixed1-23&3 &12383&0.07&294(0.976) & 107(0.95) & 92(0.925) &no/no  \\
    mixed1-24&3 &14040&0.07&413(0.979) & 150(0.958) & 125(0.934) &no/no  \\
    mixed2-10&4 &10256&0.16&362(0.984) & 130(0.969) & 113(0.954)  &no/no \\
    mixed2-11&4 &14897&0.14&989(0.988) & 384(0.98) & 278(0.965)  &no/no \\
    mixed2-12&4 &20992&0.13&2480(0.991) & 892(0.984) & 807(0.977)  &no/no \\
    mixed3-12&12 &4109&0.5&75(0.963) & 27(0.941) & 21(0.929)  &no/no \\
    mixed3-13&13&8206&0.48&554(0.982) & 198(0.973) & 171(0.968)   &no/no\\
    eco12 &12 & 1024 & 0.55 & 1(0.801) & 1(0.641) & 1(0.63) & yes/yes  \\
    sot1 &5 & 8694 & 0.01 & 21(0.745) & 9(0.62) & 9(0.552) & yes/no \\
    W1-6-5-2 & 5 & 18000 & 0.2 & 2362(0.992) & 859(0.986) & 696(0.979)& yes/yes  \\
    W1-4-6-2 & 6 & 6480 & 0.32 & 184(0.981) & 66(0.965) & 54(0.951)& yes/yes  \\
    katsura10 &11 & 1024 & 0.63 & 1(0.836) & 1(0.679) & 1(0.672)  & yes/yes
  \end{tabular}
\end{table}

We refer the reader to experiments with systems of comparable (or
higher) degrees presented in~\citep{FaMo17} (the algorithm in that
reference does not use a blocking strategy). Recall that the output
in~\citep{FaMo17} is somewhat stronger than here (the authors compute a
Gr\"obner basis of the input ideal, so multiplicities are preserved). On
the other hand, for ideals not in shape position, Table~2 of that
reference reports only the calculation of the last polynomial in the
Gr\"obner basis, whereas our algorithm makes no distinction between
ideals in shape position or not.

\section{Using the original coordinates}\label{sec:original}

In this last section, we propose a refinement of the algorithm given
previously; the main new feature is that we focus on avoiding or reducing the
use of a generic linear form $\lf = t_1 X_1 + \cdots + t_n X_n$.  Indeed, such
a linear combination is likely to result in a multiplication matrix $\mM = t_1
\mM_1 + \cdots + t_n \mM_n$ significantly less sparse than $\mM_1,\dots,\mM_n$.
In all this section, we will work under the assumption that $\mM_1$ is the
sparsest matrix among $\mM_1,\dots,\mM_n$, and try to rely on computations
involving $\mM_1$ as much as possible.

All notation, such as $I$, $\residueI=\K[X_1,\dots,X_n]/I$, its basis
$\basis=(b_1,\dots,b_D)$, the local algebras $\residueI_i$, \dots, are as in the  previous two sections.
In particular, we write $V=V(I)=\{\balpha_1,\dots,\balpha_\dg\},$ with
$\balpha_i=(\alpha_{i,1},\dots,\alpha_{i,n}) \in \Kbar{}^n$ for all $i$.

\subsection{Overview}

The algorithm decomposes $V$ into two parts: for the first part,
we will be able to use $X_1$ as a linear form in our zero-dimensional
parametrization; for the remaining points, we will use a
random linear form $\lf=t_1 X_1 + \cdots + t_n X_n$ as
above. Throughout, we rely on the following operations: evaluations of
linear forms on successive powers of a given element in $\residueI$
(such as $1,X_1,X_2^2,\dots$ or $1,\lf,\lf^2,\dots$), linear
algebra over univariate polynomials, and some
operations on univariate polynomials related to the so-called {\em
  power projection}~\citep{Shoup94,Shoup99}.

The main algorithm is as follows. For the moment, we only describe its
main structure; the subroutines are detailed in the next subsections.

\begin{algorithm}[ht]
  \caption{$\mathsf{BlockParametrizationWithSplitting}(\mM_1,\dots,\mM_n,\mU,\mV,\lf,\mf$)}
  {\bf Input:} \vspace{-0.5em}
  \begin{itemize}
    \item $\mM_1,\dots,\mM_n$ defined as above
    \item  $\mU,\mV \in \mathbb{K}^{D \times m}$, for some block dimension  $m \in \{1,\dots,D\}$
    \item $\lf =t_1 X_1 + \cdots + t_n X_n$
    \item $\mf =y_2 X_2 + \cdots + y_n X_n$
  \end{itemize}
  {\bf Output:}  \vspace{-0.5em}
  \begin{itemize}
    \item polynomials $((\sqfree,V_1,\dots,V_n),\lf)$, with $\sqfree,V_1,\dots,V_n$ in $\K[T]$
  \end{itemize}
  \begin{enumerate}
    \item let $(F,G_1,\dots,G_n,X_1)=\mathsf{BlockParametrizationX}_1(\mM_1,\dots,\mM_n,\mU,\mV,\mf)$
    \item let $(\mat{\Delta}_s),(\mat{\Delta}'_s),(\mat{\Delta}''_s)$ be correction matrices associated
      to the previous calculation
    \item let $(R,W_1,\dots,W_n,\lf)=\mainalgoname{\sf Residual}(\mM_1,\dots,\mM_n,\mU,\mV,$ \\
      \phantom{bla} \hfill $(\mat{\Delta}_s),(\mat{\Delta}'_s),(\mat{\Delta}''_s),\lf)$
    \item let $((\sqfree,V_1,\dots,V_n),\lf)=\mathsf{Union}(((F,G_1,\dots,G_n),X_1), ((R,W_1,\dots,W_n),\lf))$
    \item \textbf{return} $((\sqfree,V_1,\dots,V_n),\lf)$
  \end{enumerate}
\end{algorithm}

The call to $\mathsf{BlockParametrizationX}_1$ computes a zero-dimensional
parametrization of a subset $V'$ of $V$ such that $X_1$ separates the
points of $V'$ (that is, takes pairwise distinct values on the points
of $V'$); this is done by using sequences of the form $(\mUt
\mM_1^s \mV)_{s \ge 0}$. The next step finds three 
sequences of correction matrices, using objects computed 
in the call to  $\mathsf{BlockParametrizationX}_1$.

We then apply a modified version of Algorithm $\mainalgoname$, which we call
{\sf Block} {\sf ParametrizationResidual}. It computes a zero-dimensional
parametrization of $V''=V\setminus V'$ using sequences such as $(\mUt
\mM^s \mV)_{s \ge 0} - \mat{\Delta}_s$, where $\mM= t_1 \mM_1 + \cdots
+ t_n \mM_n$. Subtracting the correction terms $\mat{\Delta}_s$ has
the effect of removing from $V$ the points in $V'$; in those
cases where $V'$ is a large subset of $V$, then $V''=V\setminus V'$ contains
a few points, and few values for the latter sequences will be needed.
The last step involves changing coordinates in $((F,G_1,\dots,G_n),X_1)$ to use
$\lf$ as a linear form instead, and performing the union of the two components
$V'$ and $V''$. These operations can be done in time $O(n D^{(\omega+1)/2})$
\citep[Lemmas~2 \&~3]{PoSc13b}; we will not discuss them further.

\subsection{Describing the subset \texorpdfstring{$V'$}{V'} of \texorpdfstring{$V$}{V}}

In this paragraph, we give the details of Algorithm
$\mathsf{BlockParametrizationX}_1$. This is done by specializing the
discussion of \cref{ssec:genseries} to the case $\lf=X_1$:
in the notation of that section, we take
$\mathfrak{S}=\{1,\dots,\dg\}$, that is, $V_{\mathfrak{S}}=V$, and we
let $r_1,\dots,r_c$ be the pairwise distinct values taken by $X_1$ on
$V$, for some $c \le \dg$.  For $j=1,\dots,c$, we write $T_j$ for the
set of all indices $i$ in $\{1,\dots,\dg\}$ such that
$\alpha_{i,1}=r_j$; the sets $T_1,\dots,T_c$ form a partition of
$\{1,\dots,\dg\}$. When $T_j$ has cardinality~$1$, we denote it as
$T_j=\{\sigma_j\}$, for some index $\sigma_j$ in $\{1,\dots,\dg\}$, so
that $\alpha_{\sigma_j,1}=r_j$.

For $i=1,\dots,\dg$, let us write $\nu_i$ for the degree of the minimal
polynomial of $X_1$ in $\residueI_i$; thus, this polynomial is
$(T-\alpha_{i,1})^{\nu_i}$. For $j$ in $\{1,\dots,c\}$, we define
$\mu_j$ as the maximum of all $\nu_i$, for $i$ in~$T_j$. As a result, the minimal
polynomial of $X_1$ in $\prod_{j \in T_j} \residueI_j$ is 
$(T-r_j)^{\mu_j}$, and the minimal polynomial of $X_1$ in $\residueI$ is
$M=\prod_{j \in \{1,\dots,c\}} (T-r_j)^{\mu_j}$.

Recall that for any linear form $\ell: \residueI \to \K$, the
extension $\ell: \residueI\otimes_\K \Kbar \to \Kbar$ can be written
uniquely as $\ell=\sum_{i\in \{1,\dots,\dg\}} \ell_i$, with
$\ell_i:\residueI_i \to \Kbar$; collecting terms, $\ell$ may also be
written as $\ell=\sum_{j \in \{1,\dots,c\}} \lambda_j$, with
$\lambda_j=\sum_{i \in T_j} \ell_i$.  Given such an $\ell$, we first
explain how to compute values of the form $\lambda_j(1)$. We will do
this for some values of $j$ only, namely those $j$ for which
$\mu_j=1$. 

\begin{lemma}\label{lemma:valuelambda}
  Let $\ell$ be in ${\rm hom}_\K(\residueI,\K)$ and let $M$ be the minimal
  polynomial of $X_1$ in $\residueI$. Then, the polynomial
  $\Omega((\ell(X_1^s))_{s\ge0},M)$ is well-defined and satisfies
  $$\Omega((\ell(X_1^s))_{s\ge0},M)(r_j) = \lambda_j(1) M'(r_j) \quad \text{for all $j$ in $\{1,\dots,c\}$ such that $\mu_j=1$.}$$
\end{lemma}
\begin{proof}
  Let $\mathfrak{e}$ be the set of all indices $j$ in $\{1,\dots,c\}$
  such that $\mu_j=1$, and let $\mathfrak{f}=\{1,\dots,c\}-\mathfrak{e}$;
  this definition allows us to split the generating series
  of sequence $(\ell(X_1^s))_{s\ge 0}$ as
  \[
    \sum_{s \ge 0} \frac{\ell(X_1^s)}{T^{s+1}  }
  = \sum_{j \in \{1,\dots,c\}}\sum_{i\in T_j} 
  \sum_{s \ge 0} \frac{\ell_i(X_1^s)}{T^{s+1}} 
  =\sum_{j \in \mathfrak{e}}\sum_{i\in T_j}\sum_{s \ge 0}  \frac{\ell_i(X_1^s)}{T^{s+1}} +
  \sum_{j \in \mathfrak{f}}\sum_{i\in T_j}\sum_{s \ge 0}  \frac{\ell_i(X_1^s)}{T^{s+1}}.
  \]
  Using \cref{lemma:formula} with $\lf=X_1$ and $v=1$, any sum $\sum_{s \ge 0} \ell_i(X_1^s)/T^{s+1}$ 
  in the second summand
  can be rewritten as 
  $E_i/(T-r_j)^{v_i}$,
  for some integer $v_i$, and for some polynomial $E_i \in \Kbar[T]$ of degree less than
  $v_j$. Next, take $j$ in $\mathfrak{e}$. Since $\mu_j=1$, $\nu_i=1$ for all $i$ in $T_j$,
  so that
  each $\ell_i$ takes the form 
  $\ell_i: f \mapsto (\Lambda_{i}(f))(\balpha_i)$, 
  for a differential operator $\Lambda_{i}$ that does not involve $\partial/\partial
  X_1$. Since all terms of positive order in $\Lambda_i$ involve one of
  $\partial/\partial X_2,\dots,\partial/\partial X_n$, they cancel
  $X_1^s$ for $s\ge 0$. Thus, $\ell_i(X_1^s)$ can be rewritten 
  as $\ell_{i,1} \alpha_{i,1}^s$, for some constant $\ell_{i,1}$,
  and the generating series of these terms is 
  $$\frac {\ell_{i,1}}{T-\alpha_{i,1}}=\frac {\ell_{i,1}}{T-r_j}.$$
  Remarking  that we can write $\ell_{i,1}=\ell_i(1)$,
  altogether, the sum in question can be written
  \[
    \sum_{s \ge 0} \frac{\ell(X_1^s)}{T^{s+1}  }
  =\sum_{j \in \mathfrak{e}} 
  \frac{ \sum_{i\in T_j}  \ell_{i}(1) }{T-r_j }
  + \sum_{j \in \mathfrak{f}} \frac{D_j}{(T-r_j )^{x_j}}
  = \sum_{j \in \mathfrak{e}} 
  \frac{ \lambda_j(1) }{T-r_j }
  + \sum_{j \in \mathfrak{f}} \frac{D_j}{(T-r_j )^{x_j}}
  \]
  for some integers $\{x_j \mid j \in \mathfrak{f}\}$ and
  polynomials $\{D_j \mid j \in \mathfrak{f}\}$ such that
  $\deg(D_j) < x_j$ holds, and with $D_j$ and $T-r_j $
  coprime. In particular, the minimal polynomial of  $(\ell(X_1^s))_{s\ge 0}$ is $N=\prod_{j\in
  \mathfrak{e}}(T-r_j) \prod_{j \in  \mathfrak{f}}(T-r_j)^{x_j}$.

  The polynomial $N$ divides $M$, so that $x_j \le \mu_j$ holds for all $j$
  in $\mathfrak{f}$.  As a result,
  $\Omega((\ell(X_1^s))_{s\ge0},M)$ is well-defined and is given by
  \begin{align*}
    \Omega((\ell(X_1^s))_{s\ge0},M)=&
    \sum_{j \in \mathfrak{e}}
    \Big(
    \lambda_j(1) \prod_{\iota \in \mathfrak{e}-\{j\}}(T-r_\iota)\Big)
    \Big(\prod_{j \in \mathfrak{f}}(T-r_j)^{\mu_j} \Big)\\
                                    &+
                                    \sum_{j\in \mathfrak{f}}
                                    \Big(  (T-r_j)^{\mu_j-x_j} D_j
                                    \prod_{\iota \in \mathfrak{f}-\{j\}}(T-r_j)^{\mu_\iota}\Big)
                                    \Big(\prod_{j\in \mathfrak{e}} (T-r_j) \Big).
  \end{align*}
  This implies that $$\Omega((\ell(X_1^s))_{s\ge0},M)(r_k) =\lambda_k(1) 
  \prod_{\iota \in \mathfrak{e}-\{k\}}(r_k-r_\iota)
  \prod_{j \in \mathfrak{f}}(r_k-r_j)^{\mu_j} = \lambda_k(1) M'(r_k)$$ 
  holds for all $k$ in $\mathfrak{e}$.
\end{proof}

We now show how this result allows us to use sequences of the form
$(\ell(X_1^s))_{s \ge 0}$ to compute a zero-dimensional
parametrization of a subset $V'$ of $V$. Precisely, we characterize
the set $V'$ as follows: for $i$ in $\{1,\dots,\dg\}$, $\balpha_i$ is
in $V'$ if and only if:
\begin{itemize}
  \item for $i'$ in $\{1,\dots,\dg\}$, with $i'\ne i$, $\alpha_{i',1} \ne
    \alpha_{i,1}$;
  \item $\residueI_i$ is a reduced algebra, or equivalently, $I_i$ is radical (see \cref{ssec:dual} 
    for the notation $\residueI_i,I_i$).
\end{itemize}
We denote by $\mathfrak{A}\subset \{1,\dots,\dg\}$ the set of
corresponding indices $i$, and we let
$\mathfrak{B}=\{1,\dots,\dg\}\setminus\mathfrak{A}$, so that we have
$V'=V_{\mathfrak{A}}$ and $V''=V_{\mathfrak{B}}$.  Remark that $X_1$
separates the points of $V'$.

Correspondingly, we define $\mathfrak{a}$ as the set of all indices
$j$ in $\{1,\dots,c\}$ such that $\sigma_j$ is in $\mathfrak{A}$. In
other words, $j$ is in $\mathfrak{a}$ if and only if $T_j$ has
cardinality $1$, so that $T_j=\{\sigma_j\}$, and
$\residueI_{\sigma_j}$ is reduced.  The algorithm in this paragraph
will compute a zero-dimensional parametrization of $V_{\mathfrak{A}}$;
we will use the following lemma to perform the decomposition.

\begin{lemma}\label{lemma:acb2}
  Let $j$ be in $\{1,\dots,c\}$ such that $\mu_j=1$, let $\lambda$ be
  a linear form over $\prod_{i \in T_j} \residueI_i$ and let $\mf=y_2
  X_2 + \cdots + y_n X_n$, for some $y_2,\dots,y_n$ in $\K$. Define constants
  $a=\lambda(1)$, $b=\lambda(\mf)$, $c=\lambda(\mf^2)$ in $\Kbar$.
  Then, $j$ is in $\mathfrak{a}$
  if and only if, for a generic choice of $\lambda$ and $\mf$, $ac=b^2$.
\end{lemma}
\begin{proof}
  The assumption that $\mu_j=1$ means that for all $i$ in $T_j$,
  $\nu_i=1$. The linear 
  form $\lambda$ can be uniquely written as a sum $\lambda=\sum_{i \in T_j}
  \ell_i$, where each $\ell_i$ is in ${\rm hom}_\Kbar(\residueI_i,\Kbar)$.
  The fact that all $\nu_i$ are equal to $1$ then implies that each $\ell_i$ takes the form 
  $$\ell_i: f \mapsto (\Lambda_{i}(f))(\balpha_i),$$
  where $\Lambda_{i}$ is a differential operator that does not 
  involve $\partial/\partial X_1$. Thus, as in \cref{ell_param}, we can write a general
  $\Lambda_i$ of this form as
  $$\Lambda_i: f \mapsto u_{i,1} f + \sum_{2 \le r \le n}
  P_{i,r}(u_{i,2},\dots,u_{i,D_i}) \frac{\partial}{\partial X_j} f +
  \sum_{2 \le r \le s \le n} P_{i,r,s}(u_{i,2},\dots,u_{i,D_i})
  \frac{\partial^2}{\partial X_j\partial X_k} f +
  \tilde\Lambda_i(f),$$ where all terms in $\tilde \Lambda_i$ have
  order at least $3$, $\bu_i=(u_{i,1},\dots,u_{i,D_i})$ are parameters and
  $(P_{i,r})_{2 \le r \le n}$ and $(P_{i,r,s})_{2 \le r \le s \le n}$
  are linear forms in $u_{i,2},\dots,u_{i,D_i}$.
  We obtain
  \begin{align*}
    \Lambda_i(1)   &= u_{i,1} \\
    \Lambda_i(\mf)   &= u_{i,1} \mf +\sum_{2 \le r \le n}P_{i,r}(u_{i,2},\dots,u_{i,D_i})y_r \\
    \Lambda_i(\mf^2) &= u_{i,1} \mf^2  +2 \mf \sum_{2 \le r \le n}P_{i,r}(u_{i,2},\dots,u_{i,D_i})y_r + 
    2\sum_{2 \le r \le s \le n} P_{i,r,s}(u_{i,2},\dots,u_{i,D_i})y_ry_s,
  \end{align*}
  which gives
  \begin{align*}
    a &= \sum_{i\in T_j}u_{i,1} \\
    b &= \sum_{i\in T_j}u_{i,1} \mf(\balpha_i) +\sum_{i \in T_j, 2 \le r \le n}P_{i,r}(u_{i,2},\dots,u_{i,D_i})y_r \\
    c &= \sum_{i\in T_j}u_{i,1} \mf(\balpha_i)^2
        +2 \sum_{i \in T_j, 2 \le r \le n}\mf(\balpha_i) P_{i,r}(u_{i,2},\dots,u_{i,D_i})y_r    
        +2 \sum_{i \in T_j, 2 \le r \le s \le n} P_{i,r,s}(u_{i,2},\dots,u_{i,D_i})y_ry_s.
  \end{align*}
  Suppose first that $j$ is in $\mathfrak{a}$. Then, $T_j=\{\sigma_j\}$, so we 
  have only one term $\Lambda_{\sigma_j}$ to consider, and $\residueI_{\sigma_j}$ 
  is reduced, so that all coefficients $P_{\sigma_j,r}$ and
  $P_{\sigma_j,r,s}$ vanish. Thus, we are left in
  this case with
  $$
  a = u_{\sigma_j,1}, \quad
  b = u_{\sigma_j,1} \mf(\balpha_{\sigma_j}), \quad
  c = u_{\sigma_j,1} \mf(\balpha_{\sigma_j})^2,
  $$ so that we have $ac=b^2$, for {\em any} choice of $\lambda$ and
  $\mf$. Now, we suppose that $j$ is not in $\mathfrak{a}$, and we prove
  that for a generic choice of $\lambda$ and $\mf$, $ac-b^2$ is nonzero.
  The quantity $ac-b^2$ is a polynomial in the parameters
  $(\bu_i)_{i\in T_j}$, and $(y_i)_{i \in \{2,\dots,n\}}$, and we have
  to show that it is not identically zero. We discuss two cases; in both
  of them, we prove that a suitable specialization of $ac-b^2$ is
  nonzero.

  Suppose first that for at least one index $\sigma$ in $T_j$,
  $\residueI_\sigma$ is not reduced. In this case, there exists as least one
  index $\rho$ in $\{2,\dots,n\}$ such that
  $P_{\sigma,\rho}(u_{\sigma,2},\dots,u_{\sigma,D_\sigma})$ is not
  identically zero. Let us set all $\bu_{\sigma'}$
  to $0$, for $\sigma'$ in $T_j-\{\sigma\}$, as well as $u_{\sigma,1}$,
  and all $y_r$ for $r\ne \rho$. Then, under this specialization,
  $ac-b^2$ becomes
  $-(P_{\sigma,\rho}(u_{\sigma,2},\dots,u_{\sigma,D_\sigma})y_\rho)^2$
  which is nonzero, hence $ac-b^2$ itself is nonzero.

  Else, since $j$ is not in $\mathfrak{a}$, we can assume that $T_j$
  has cardinality at least $2$, with $\residueI_\sigma$ reduced for all $\sigma$
  in $T_j$ (so that $P_{\sigma,r}$ and $P_{\sigma,r,s}$ vanish for 
  all such $\sigma$ and all $r,s$). Suppose that $\sigma$ and $\sigma'$ are two indices in
  $T_j$; we set all indices $u_{\sigma'',1}$ to zero, for $\sigma''$
  in $T_j-\{\sigma,\sigma'\}$. We are left with
  $$
  a=u_{\sigma,1}+u_{\sigma',1},\quad
  b=u_{\sigma,1}\mf(\balpha_{\sigma})+u_{\sigma',1}\mf(\balpha_{\sigma'}),\quad
  c=u_{\sigma,1}\mf(\balpha_{\sigma})^2+u_{\sigma',1}\mf(\balpha_{\sigma'})^2.
  $$
  Then, $ac-b^2$ is equal to $2u_{\sigma,1}u_{\sigma',1}(\mf(\balpha_{\sigma})-\mf(\balpha_{\sigma'}))^2$,
  which is nonzero, since $\balpha_\sigma \ne \balpha_{\sigma'}$.
\end{proof}

The previous lemmas allow us to write Algorithm
$\mathsf{BlockParametrizationX}_1$. After computing $M$, we determine its
factor $F=\prod_{j \in \{1,\dots,c\}, \mu_j=1} (T-r_j)$, which is
obtained by taking the squarefree part of $M$ and dividing it by
its gcd with $\gcd(M,M')$. We split this polynomial further using the previous
lemma in order to find $\prod_{j \in \mathfrak{a}} (T-r_j)$, and we
conclude using the same kind of calculations as in Algorithm
$\mainalgoname$.

\begin{algorithm}[ht]
  \caption{$\mainalgoname{\sf X}_1(\mM_1,\dots,\mM_n,\mU,\mV,\mf$)}
  {\bf Input:} \vspace{-0.5em}
  \begin{itemize}
    \item multiplication matrices $\mM_1,\dots,\mM_n$ in $\K^{D \times D}$
    \item  $\mU,\mV \in \mathbb{K}^{D \times m}$, for some block dimension  $m \in \{1,\dots,D\}$
    \item $\mf =y_2 X_2 + \cdots + y_n X_n$
  \end{itemize}
  {\bf Output:}  \vspace{-0.5em}
  \begin{itemize}
    \item polynomials $((F,G_1,\dots,G_n),X_1)$, with $F,G_1,\dots,G_n$ in $\K[T]$
  \end{itemize}
  \begin{enumerate}
    \item\label{X1step3} { compute $\mat{L}_s = \mUt\mM_1^s$ for $s=0,\dots,2d-1$, with $d = \lceil D/m \rceil$}
    \item\label{X1step4} { compute $\seqelt{s,\mU,\mV}= \mat{L}_s\mV$ for $s=0,\dots, 2d-1$}
    \item\label{X1step5} { compute a minimal matrix generator $\mat{P}_{\mU,\mV}$ of $(\seqelt{s,\mU,\mV})_{0 \le s < 2d}$}
    \item\label{X1step6} { let $M$ be the largest invariant factor of $\mat{P}_{\mU,\mV}$}
    \item\label{X1step7} { let $F$ be  the squarefree part  of $M$}
    \item\label{X1step7b} let $F = F /\gcd(F, \gcd (M, M'))$
    \item\label{X1step8} { let $\row{a}_1 = [M~0 ~\cdots~ 0] (\mat{P}_{\mU,\mV})^{-1}$}
    \item let $\mN = y_2 \mM_2 + \cdots+y_n \mM_n$
    \item \textbf{for} $i=0,1,2$ \textbf{do} \\
      \phantom{for}  let $A_i = \mathsf{{\sf ScalarNumerator}}(\mat{P}_{\mU,\mV}, M, 
      \mN^i \col{\varepsilon}_1, 1, \row{a}_1, (\mat{L}_s)_{0 \le s<d})$
    \item\label{X1step9} let $F = \gcd(F, A_0A_2-A_1^2)$
    \item\label{X1step10} \textbf{for} $i=2,\dots,n$ \textbf{do} \\
      \phantom{for}  let $A_{X_i} = \mathsf{{\sf ScalarNumerator}}(\mat{P}_{\mU,\mV}, M, \mM_i\col{\varepsilon}_1, 1, \row{a}_1, (\mat{L}_s)_{0 \le s <d})$
    \item\label{X1step11}   \textbf{return} $((F,T,A_{X_2}/ A_0 \bmod F, \dots,A_{X_n}/A_{0} \bmod F),X_1)$
  \end{enumerate}  \label{algo:X1block-sparse-fglm}
\end{algorithm}

\begin{lemma}
  For generic choices of $\mU$, $\mV$ and $\mf$, the output
  $((F,G_1,\dots,G_n),X_1)$ of the algorithm $\mathsf{BlockParametrizationX}_1$
  is a zero-dimensional parametrization of $V_{\mathfrak{A}}$.
\end{lemma}
\begin{proof}
  As in the case of $\mainalgoname$, for generic choices of $\mU$ and
  $\mV$, the degree bound $\deg(\mat{P}_{\mU,\mV}) \le d$ holds and
  $M$ is the minimal polynomial of $X_1$ in $\residueI$; hence, after
  Step~\ref{X1step7b}, we have $F=\prod_{j \in \{1,\dots,c\}, \mu_j=1}
  (T-r_j)$. 

  The calculation of $A_0,A_1,A_2$ and $A_{X_2},\dots,A_{X_n}$ is
  justified as in $\mainalgoname$, by means of
  \cref{lemma:omegaOmega}. For $i=1,\dots,D$, let $u_{i,1}$ is the entry at position
  $(i,1)$ in $\mU$, and define the linear 
  form $\ell: \residueI \to \K$ by 
  \[
    \ell(f) = \sum_{i=1}^D f_i u_{i,1}, \quad\text{for}\quad f =
    \sum_{i=1}^D f_i b_i.
  \]
  Then, the above lemma proves that we have $A_i = \Omega((\ell(\mf^i
  X_1^s))_{s\ge0},M)$ for $i=0,1,2$ as well as $A_{X_i} = \Omega((\ell(X_i
  X_1^s))_{s\ge0},M)$ for $i=2,\dots,n$.

  Take then $j$ in $\{1,\dots,c\}$ such that $\mu_j=1$, that is, a
  root of $F$ as computed at Step~\ref{X1step7b}. By
  \cref{lemma:valuelambda}, for $i=0,1,2$, we have $ A_i(r_j) = M'(r_j)
  (\mf^i \cdot \lambda_j)(1)$, where $\lambda_j =\sum_{i \in T_j}
  \ell_i$, where the $\ell_i$'s are the components of $\ell$,
  and where $\mf^i \cdot \lambda_j$ is the linear form $f \mapsto \lambda_j(\mf^i f)$.

  As a result, the value of $ A_0  A_2 -  A_1^2$ at
  $r_j$ is (up to the nonzero factor $M'(r_j)^2$) equal to the
  quantity $ac-b^2$ defined in \cref{lemma:acb2}, so for a
  generic choice of $\ell$ (that is, of $\mU$) and $\mf$, it vanishes if and only if $j$ is
  in $\mathfrak{a}$. Thus, after Step~\ref{X1step9}, 
  $F$ is equal to $\prod_{j \in \mathfrak{a}} (T-r_j)$.

  The last step is to compute the zero-dimensional parametrization of
  $V_{\mathfrak{A}}$. This is done using again
  \cref{lemma:valuelambda}. Indeed, for $j$ in $\mathfrak{a}$, 
  $T_j$ is simply equal to $\{\sigma_j\}$, so that we have, for $i=2,\dots,n$,
  $$ A_0(r_j)=M'(r_j) \lambda_j(1) \quad\text{and}\quad 
  A_{X_i}(r_j) = M'(r_j) (X_i \cdot \lambda_j)(1) = M'(r_j) \lambda_j(X_i),$$
  where as above, $X_i \cdot \lambda_j$ is the linear form $f \mapsto \lambda_j(X_i f)$.
  Now, since $j$
  is in $\mathfrak{a}$, $\residueI_{\sigma_j}$ is reduced, so that there
  exists a constant $\lambda_{j,1}$ such that for all $f$ in
  $\Kbar[X_1,\dots,X_n]$, $\lambda_j(f)$ takes the form $\lambda_{j,1}
  f(\balpha_{\sigma_j})$. This shows that, as claimed,
  $$\frac{ A_{X_j}(r_j)}{ A_0 (r_j)} = 
  \frac
  {M'(r_j) \lambda_{j,1} \alpha_{\sigma_j,i}}{M'(r_j) \lambda_{j,1}} = \alpha_{\sigma_j,i},$$
  since $M'(r_j)$ is nonzero.
  For $i=1$, since we use $X_1$ as a separating variable for $V_{\mathfrak{A}}$, 
  we simply insert the polynomial $T$ into our list.
\end{proof}

The cost analysis is the same as that of Algorithm $\mainalgoname$, the
crucial difference being that the density $\rho_1$ of $\mM_1$ plays the 
role of the density $\rho$ of $\mM$ used in that algorithm.

\subsection{Computing correction matrices}

Next, we describe an operation of decomposition for linear forms
$\residueI \to \K$; this is essentially akin to the Chinese Remainder
Theorem. We then use it to compute the sequences of correction
matrices $(\mat{\Delta}_s),(\mat{\Delta}'_s),(\mat{\Delta}''_s)$
defined in Algorithm $\mathsf{BlockParametrizationWithSplitting}$.

As a preamble, we introduce the notation $\PP(r,t)$ for the cost of a
{\em power projection} operation, as defined
in~\citep{Shoup94,Shoup99}: given a polynomial $F$ in $\K[T]$ of degree
$r$, a linear form $\ell: \K[T]/F \to \K$, and $H$ in $\K[T]/F$, the
goal is to compute $(\ell(H^s))_{0 \le s < t}$, for some upper bound
$t$. We denote this operation by ${\sf PowerProjection}(F,H,\ell,t)$; this is
essentially the analogue for univariate polynomials of the Krylov
computations that we heavily rely on in this paper. Here, $\ell$ 
is represented by the vector $(\ell(1),\ell(T),\dots,\ell(T^{r-1}))$.

\citet[Theorem~4]{Shoup94} showed that this can be done in
$\PP(r,t)=O(r^{(\omega-1)/2} t)$ operations in $\K$ for $t \le r$. For
$t \ge r$, we first solve the problem up to index $r$ in time
$O(r^{(\omega+1)/2})$; then we use the fact that the sequence 
$(\ell(H^s))_{s \ge 0}$
is linearly recurrent to compute all further values in time
$O(t\M(r)/r)$, as for instance in~\citep[Proposition~1]{BoFlSaSc06}.
Thus, for $t \ge r$, we take $\PP(r,t)=O(r^{(\omega+1)/2} +
t\M(r)/r)$.

\medskip

Let $\mathfrak{A}$ and $\mathfrak{B}$ be defined as in the previous 
subsection, and let $D_{\mathfrak{A}}$ be the number of points in
$V_{\mathfrak{A}}$. Since $\residueI_i$ is a reduced algebra for all
$i$ in $\mathfrak{A}$, $D_\mathfrak{A}$ is also the dimension of
$\residueI_\mathfrak{A} = \prod_{i \in \mathfrak{A}}
\residueI_i$ and $\residueI_\mathfrak{B}=\prod_{i \in \mathfrak{B}}
\residueI_i$ has dimension $D_{\mathfrak{B}}=D-D_{\mathfrak{A}}$.

Consider a linear form $\ell: \residueI \to \K$; we still denote
$\ell$ for its extension $\residueI \otimes_\K \Kbar \to \Kbar$.  It
can then be decomposed as $\ell= \ell_{\mathfrak{A}} +
\ell_{\mathfrak{B}}$, with $\ell_\mathfrak{A}: \residueI_\mathfrak{A}
\to \Kbar$ and $\ell_\mathfrak{B}: \residueI_\mathfrak{B} \to \Kbar$.
Remark that the support of $\ell_\mathfrak{B}$ is contained in
$\mathfrak{B}$, and actually equal to $\mathfrak{B}$ for a generic
$\ell$.

Suppose that we are given the minimal polynomial $M$ of $X_1$ in
$\residueI$, the numerator $C=\Omega( (\ell(X_1^s))_{s \ge 0}, M)$, as
well as the zero-dimensional parametrization $((F,G_1,\dots,G_n),X_1)$
of $V_\mathfrak{A}$ computed in the previous paragraph.  Given
$\lf=t_1 X_1 + \cdots+ t_n X_n$, and an upper bound $\tau$, we show
how to compute the values $\ell_\mathfrak{A}(\lf^s)$, for
$s=0,\dots,\tau-1$.

Let $E=M/F$; the division is exact and  $E$ and $F$ are coprime, by construction.
The equality $\ell= \ell_{\mathfrak{A}} + \ell_{\mathfrak{B}}$ implies
an equality between generating series
\[
  \sum_{s \ge 0} \frac{\ell(X_1^s)}{T^{s+1}} = \sum_{s \ge 0}\frac{\ell_\mathfrak{A}(X_1^s)}{T^{s+1}}  
                                              +\sum_{s \ge 0} \frac{\ell_\mathfrak{B}(X_1^s)}{T^{s+1}}
                                             = \frac{A}{F} + \frac{B}{E},
\]
for some polynomials $A,B$ in $\K[T]$, with $\deg(A) < \deg(F)$ and
$\deg(B) < \deg(E)$. With  $C=\Omega( (\ell(X_1^s))_{s \ge 0}, M)$, we deduce the equality
$$\frac{C}{M}=\frac{A}{F} + \frac{B}{E},$$ from which we find $A = C/E
\bmod F$. Knowing $A$ and $F$ allows us to compute the values
$\ell_\mathfrak{A}(X_1^s)$, for $s=0,\dots,D_\mathfrak{A}-1$, by
Laurent series expansion.  Since $\residueI_\mathfrak{A}$ is reduced,
we have $\lf = t_1 G_1 + \cdots t_n G_n$ in $\residueI_\mathfrak{A}$,
where $G_1,\dots,G_n$ are polynomials in the zero-dimensional
parametrization of $V_\mathfrak{A}$. As a result, we can finally
compute $\ell_\mathfrak{A}(\lf^s)$, for $s=0,\dots,\tau-1$ by
applying our algorithm for univariate power projection to $G=t_1 G_1 + \cdots t_n G_n$.

\cref{algo:decompose} (\textsf{Decompose}) summarizes this discussion. Its cost
bound is
\[
  O(\M(D_\mathfrak{A})\log(D_\mathfrak{A}) + \PP(D_\mathfrak{A},\tau) +nD_\mathfrak{A})
\]
operations in $\K$, where the first term accounts for the cost of the first
three steps, $\PP(D_\mathfrak{A},\tau)$ is the cost of power projection and the
term $O(nD_\mathfrak{A})$ is the cost of computing $G$ as defined above.

\begin{algorithm}[ht]
  \caption{${\sf Decompose}(M, C, ((F,G_1,\dots,G_n),X_1), \lf,\tau)$} {\bf
  Input:} \vspace{-0.5em}
  \begin{itemize}
    \item minimal polynomial $M$ of $X_1$ in $\residueI$
    \item numerator $C=\Omega( (\ell(X_1^s))_{s \ge 0}, M)$
    \item zero-dimensional parametrization $((F,G_1,\dots,G_n),X_1)$ of $V_\mathfrak{A}$
    \item $\lf =t_1 X_1 + \cdots + t_n X_n$
    \item a bound $\tau$
  \end{itemize}
  {\bf Output:}  \vspace{-0.5em}
  \begin{itemize}
    \item $\ell_\mathfrak{A}(\lf^s)$, for $s=0,\dots,\tau-1$
  \end{itemize}
  \begin{enumerate}
    \item let $E=M/F$
    \item let $A=C/E \bmod F$
    \item compute the first $D_\mathfrak{A}$ terms $(v_0,\dots,v_{D_\mathfrak{A}-1})$ of the Laurent series $A/F$
    \item {\bf return} ${\sf PowerProjection}(F, t_1 G_1 + \cdots + t_n G_n, (v_0,\dots,v_{D_\mathfrak{A}-1}),\tau)$
  \end{enumerate}
  \label{algo:decompose}
\end{algorithm}

Algorithm \textsf{Decompose} will be used for obtaining correction matrices
given as input of Algorithm $\mainalgoname{\sf Residual}$. We assume that we
have stored various quantities computed in Algorithm $\mainalgoname{\sf X}_1$:
the sequence of matrices $(\seqelt{s,\mU,\mV})_{0 \le s < 2d}$, the matrix
generator $\mat{P}_{\mU,\mV}$, the minimal polynomial $M$ of $X_1$, and the
parametrization $((F,G_1,\dots,G_n),X_1)$.

Let $\mU$ and $\mV$ be the blocking matrices used in
$\mainalgoname{\sf X}_1$. For $i=1,\dots,m$, we let
$\ell_i: \residueI \to \K$ be the linear form whose values on the
basis $\basis=(b_1,\dots,b_D)$ are given by the $i$-th column of
$\mU$. In other words, $\ell_i(f) = \sum_{j=1}^D f_j u_{j,i}$, for
$f=\sum_{j=1}^D f_j b_j$. Similarly, for $j=1,\dots,m$, we let
$\gamma_j$ be the element of $\residueI$ whose coefficient vector on
the basis $\basis$ is the $j$-th column of $\mV$.
Hereafter, we write $d_\mathfrak{B}=\lceil D_\mathfrak{B}/m \rceil$,
in analogy with the definition of $d$ used so far.
\begin{itemize}
  \item To each $(i,j)$ in $\{1,\dots,m\}\times \{1,\dots,m\}$ is
    associated a linear form $\ell_{i,j}: \residueI\to \K$ defined by
    $\ell_{i,j}(f) =\ell_i(\gamma_j f)$ for all $f$ in $\residueI$.
    Then, the entry $(i,j)$ of the matrix sequence $(\mUt \mM_1^s
    \mV)_{s \ge0}$ is the scalar sequence $(\ell_{i,j}(X_1^s))_{s \ge
    0}$.

    \smallskip

    For all such $(i,j)$, since we know the minimal polynomial $M$ of
    $X_1$, we can compute the scalar numerator $C_{i,j} \in \K[T]$
    associated to $\ell_{i,j}$ and $M$. This is done by applying
    Algorithm $\mathsf{{\sf ScalarNumerator}}$ of
    \cref{ssec:scalar_numer}, using the row vector $\row{a}_i$ defined
    in that section, together with the sequence of matrices
    $\seqelt{s,\mU,\mV}$ and the matrix generator $\mat{P}_{\mU,\mV}$
    computed in Algorithm $\mainalgoname{\sf X}_1$.

    \smallskip

    Once $C_{i,j}$ is known, we can call ${\sf Decompose}$, which allows us
    to compute $\ell_{i,j,\mathfrak{A}}(\lf^s)$, for
    $s=0,\dots,2d_\mathfrak{B}-1$.  We can then construct the sequence
    $(\mat{\Delta}_s)_{0 \le s < 2d_\mathfrak{B}}$ of matrices in
    $\K^{m\times m}$ by setting the $(i,j)$-th entry of $\mat{\Delta}_s$
    to be $\ell_{i,j,\mathfrak{A}}(\lf^s)$.

    \smallskip

  \item To each $(i,k)$ in $\{1,\dots,m\}\times \{1,\dots,n\}$ is
    associated a linear form $\ell'_{i,k}: \residueI\to \K$ defined by
    $\ell'_{i,k}(f) =\ell_i(X_k f)$ for all $f$ in $\residueI$.  Then,
    the $i$th entry of the sequence of column vectors $(\mUt
    \mM_1^s \mM_k \col{\varepsilon}_1)_{s \ge0}$ is the scalar sequence
    $(\ell'_{i,k}(X_1^s))_{s \ge 0}$, where $\col{\varepsilon}_1$ is the
    column vector $\trsp{[1~0\cdots~0]}$ we already used several times.

    \smallskip

    Proceeding as before, we construct the sequence of $m \times n$ matrices
    $(\mat{\Delta}'_s)_{0 \le s < d_\mathfrak{B}}$ by setting the
    $(i,k)$-th entry of $\mat{\Delta}'_s$ to be
    $\ell'_{i,k,\mathfrak{A}}(\lf^s)$. Note that we will only need $d_\mathfrak{B}$
    entries in this sequence.

    \smallskip

  \item Finally, we apply this process to the linear forms $\ell_i$
    themselves; they are such that the $i$th entry of the sequence of
    column vectors $(\mUt \mM_1^s \col{\varepsilon}_1)_{s \ge0}$ is
    the scalar sequence $(\ell_{i}(X_1^s))_{s \ge 0}$. Using again
    $\mathsf{{\sf ScalarNumerator}}$ and {\sf Decompose}, we construct
    the sequence of column vectors $(\mat{\Delta}''_s)_{0 \le s <
    d_\mathfrak{B}}$ by setting the $i$-th entry of $\mat{\Delta}''_s$
    to $\ell_{i,\mathfrak{A}}(\lf^s)$.
\end{itemize}

In terms of cost, computing the vectors $\row{a}_i$, for
$i=1,\dots,m$, uses $O(m^{\omega} \M(D) \log(D) \log(m))$ operations in $\K$
(see \cref{eqn:hol_cost}).  Then, the total time
spent in {\sf ScalarNumerator} is $m(m+n+1)$ times the cost reported in
\cref{ssec:scalar_numer}, which was $O(D^2 + m\M(D))$; similarly, the total
cost incurred by ${\sf Decompose}$ is $m(m+n+1)$ times the cost of a single
call, which was reported above.

\subsection{Describing the residual set}

We finally describe Algorithm $\mainalgoname{\sf Residual}$.  Let
$\mathfrak{A}$ and $\mathfrak{B}$ be as in the previous section.  This
part of the main algorithm computes a zero-dimensional parametrization
of the residual set $V_\mathfrak{B}=V\setminus V_\mathfrak{A}$. For this, we
are going to call a modified version Algorithm
$\mathsf{BlockParametrization}$, where we update the values of our
matrix sequences before computing the minimal matrix generator, using
the correction matrices defined just above.

The resulting algorithm is as follows. A superficial difference with
$\mathsf{BlockParametrization}$ is that names of the main variables
have been changed (so as not to create any confusion with those used
in $\mainalgoname{\sf X}_1$). More importantly, using the correction
matrices makes it possible for us to compute fewer terms in the
sequences, namely only $2\lceil D_\mathfrak{B}/m\rceil$ and
$\lceil D_\mathfrak{B}/m\rceil$, respectively. Hence, if $D_\mathfrak{B} \ll D$
(that is, $V_\mathfrak{B}$ contains few points, with small
multiplicities), this last stage of the algorithm will be fast.

The algorithm uses a subroutine called {\sf ScalarNumeratorCorrected}
at Steps~8 and~9.  It is similar to Algorithm {\sf ScalarNumerator} of
\cref{ssec:scalar_numer}, with a minor difference: instead of computing the
vectors $\mat{D}_s \col{\varepsilon}_1$, resp.\ $\mat{D}_s\mM_i
\col{\varepsilon}_1$, at the first step of  {\sf ScalarNumerator}, it computes $\mat{D}_s
\col{\varepsilon}_1-\mat{\Delta}''_s$, resp.\ $\mat{D}_s\mM_i
\col{\varepsilon}_1-\mat{\Delta}'_{s,i}$, where $\mat{\Delta}'_{s,i}$
is the $i$th column of $\mat{\Delta}'_{s}$.

\begin{algorithm}[ht]
  \caption{$\mainalgoname{\sf Residual}(\mM_1,\dots,\mM_n,\mU,\mV,(\mat{\Delta}_s),(\mat{\Delta}'_s),(\mat{\Delta}''_s),\lf)$}
  {\bf Input:} \vspace{-0.5em}
  \begin{itemize}
    \item $\mM_1,\dots,\mM_n$ defined as above
    \item  $\mU,\mV \in \mathbb{K}^{D \times m}$, for some block dimension  $m \in \{1,\dots,D\}$
    \item sequences of correction matrices $(\mat{\Delta}_s),(\mat{\Delta}'_s),(\mat{\Delta}''_s)$
    \item $\lf =t_1 X_1 + \cdots + t_n X_n$
  \end{itemize}
  {\bf Output:}  \vspace{-0.5em}
  \begin{itemize}
    \item  polynomials $((R,W_1,\dots,W_n),\lf)$, with $R,W_1,\dots,W_n$ in $\K[T]$
  \end{itemize}
  \begin{enumerate}
    \item\label{residualstep1}   let $\mM = t_1 \mM_1 + \cdots + t_n \mM_n$
    \item\label{residualstep3} { compute $\mat{D}_s = \mUt\mM^s$ for $s=0,\dots,2d_\mathfrak{B}-1$, with $d_\mathfrak{B} = \lceil D_\mathfrak{B}/m \rceil$}
    \item\label{residualstep4} { compute $\mat{E}_{s,\mU,\mV}= \mat{D}_s\mV-\mat{\Delta}_s$ for $s=0,\dots, 2d_\mathfrak{B}-1$}
    \item\label{residualstep5} { compute a minimal matrix generator $\mat{S}_{\mU,\mV}$ of $(\mat{E}_{s,\mU,\mV})_{0 \le s < 2d_\mathfrak{B}}$}
    \item\label{residualstep6} { let $S$ be the largest invariant factor of $\mat{S}_{\mU,\mV}$}
    \item\label{residualstep7} { let $R$ be  the squarefree part  of $S$}
    \item\label{residualstep8} { let $\row{a}_1 = [S~0 ~\cdots~ 0] (\mat{S}_{\mU,\mV})^{-1}$}
    \item\label{residualstep9}  let $C_1 = \mathsf{{\sf ScalarNumeratorCorrected}}(\mat{S}_{\mU,\mV}, S, \col{\varepsilon}_1, 1, \row{a}_1,  (\mat{D}_s)_{0 \le s < d_\mathfrak{B}}, (\mat{\Delta}''_s)_{0 \le s < d_\mathfrak{B}})$
    \item\label{residualstep10} \textbf{for} $i=1,\dots,n$ \textbf{do} \\
      \phantom{for}let $C_{X_i} = \mathsf{{\sf ScalarNumeratorCorrected}}(\mat{S}_{\mU,\mV},S, \mM_i\col{\varepsilon}_1, 1, \row{a}_1, (\mat{D}_s)_{0 \le s < d_\mathfrak{B}}, (\mat{\Delta}'_s)_{0 \le s < d_\mathfrak{B}})$
    \item\label{residualstep11}     \textbf{return} $((R, C_{X_1}/ C_1 \bmod R, \dots, C_{X_n}/ C_{1} \bmod R),\lf)$
  \end{enumerate}  \label{algo:block-sparse-fglm-residual}
\end{algorithm}

Let us prove correctness. Since $\residueI=\residueI_\mathfrak{A}
\times \residueI_\mathfrak{B}$, we may assume without loss of
generality that our multiplication matrices are block diagonal, with
two blocks corresponding respectively to bases of $\residueI_\mathfrak{A}$
and $\residueI_\mathfrak{B}$; if not, apply a change of basis to 
reduce to this situation, updating $\mU$ and $\mV$ accordingly. 

We denote by $\mM_\mathfrak{A}$ and $\mM_\mathfrak{B}$ the 
two blocks on the diagonal of matrix $\mM$.
The projection matrices can also be divided into blocks, namely as
$$\mU = \left [\begin{matrix}\mU_\mathfrak{A} \\\mU_\mathfrak{B}
\end{matrix}\right ] \quad\text{and}\quad
\mV = \left [\begin{matrix}\mV_\mathfrak{A} \\\mV_\mathfrak{B}
\end{matrix}\right ],$$
and we have $\mUt \mM^s \mV = \mUt_\mathfrak{A}
\mM_\mathfrak{A}^s \mV_\mathfrak{A} + \mUt_\mathfrak{B}
\mM_\mathfrak{B}^s \mV_\mathfrak{B}$ for $s \ge 0$. The first summand
is none other than the matrix $\mat{\Delta}_s$, so that
$\mat{E}_{s,\mU,\mV}$ is equal to $\mUt_\mathfrak{B}
\mM_\mathfrak{B}^s \mV_\mathfrak{B}$. These are thus the kind of
Krylov matrices we would obtain if we were working with a basis of
$\residueI_\mathfrak{B}$, and shows that we have enough
terms to compute a minimal matrix generator
$\mat{S}_{\mU,\mV}$, at least for generic $\mU$ and $\mV$. 
Similarly, $S$ is generically the minimal polynomial of $X_1$ in
$Q_\mathfrak{B}$, and $R$ its squarefree part.

The same considerations justify the computation of $C_1$ and
$C_{X_1},\dots,C_{X_n}$. Indeed, subtracting the correction matrices
implies
\begin{align*}
  C_1 &= \Omega( (\ell_1(\lf^s)-\ell_{1,\mathfrak{A}}(\lf^s))_{s \ge 0}, S) =  \Omega( (\ell_{1,\mathfrak{B}}(\lf^s))_{s \ge 0}, S), \\
  C_{X_i} &= \Omega( (\ell_1(X_i \lf^s)-\ell_{1,\mathfrak{A}}(X_i \lf^s))_{s \ge 0}, S) =  \Omega( (\ell_{1,\mathfrak{B}}(X_i\lf^s))_{s \ge 0}, S) \;\;\;\text{for } i=1,\dots,n.
\end{align*}
As shown in \cref{ssec:abstractlago}, these are precisely 
the polynomials we need in order to compute a zero-dimensional parametrization
of $V_\mathfrak{B}$.

The cost analysis is similar to that of Algorithm $\mainalgoname$,
with the important exception that the sequence length $d=\lceil
D/m\rceil$ can then be replaced by $d_\mathfrak{B}=\lceil
D_\mathfrak{B}/m\rceil$  (which is hopefully much smaller).

\subsection{Experimental results}

The algorithms in this section were implemented using the same
framework as in the previous section, using in particular NTL's
built-in implementation of power projection.

In \cref{tbl:comparison_algos}, we give the ratio of the runtime of
{\sf Block\-Parametrization\-WithSplitting} to that of our first
algorithm, $\mainalgoname$, for each input: numbers less than $1$
indicate a speed-up.  The last column shows the number of points in
$V_\mathfrak{A}$, that is, $D_\mathfrak{A}$, compared to the total
degree of $I$, which is $D=D_\mathfrak{A}+D_\mathfrak{B}$. The inputs,
the machine used for timings and the prime field are the same as in
\cref{section:ex}.

\begin{table}[ht]
  \centering
  \setlength\tabcolsep{6pt}
  \caption{Comparison of $\mathsf{BlockParametrizationWithSplitting}$ and $\mainalgoname$}
  \label{tbl:comparison_algos}
  \begin{tabular}{c|c|c|c|c|c|c}
    \textbf{name}& $\bm{n}$ & $\bm{D}$ & $\bm{m = 1}$ & $\bm{m = 3}$ & $\bm{m = 6}$&$D_\mathfrak{A}/D$\\
    \hline
    rand1-26&3 &17576&0.453&0.384&0.65&17576/17576 \\
    rand1-28&3 &21952&0.438&0.435&0.562& 21952/21952\\
    rand1-30&3 &27000&0.429&0.577&0.608&27000/27000 \\
    rand2-10&4 &10000&0.437&0.462&0.49& 10000/10000\\
    rand2-11&4 &14641&0.423&0.566&0.435&14641/14641 \\
    rand2-12&4 &20736&0.431&0.437&0.399&20736/20736 \\
    mixed1-22&3 &10864&0.49&0.568&0.791& 10648/10675\\
    mixed1-23&3 &12383&0.477&0.546&0.655& 12167/12194\\
    mixed1-24&3 &14040&0.463&0.514&0.613& 13824/13851\\
    mixed2-10&4 &10256&0.43&0.482&0.626& 10000/10016\\
    mixed2-11&4 &14897&0.414&0.408&0.521& 14641/14657\\
    mixed2-12&4 &20992&0.416&0.438&0.416&20736/20752 \\
    mixed3-12&12 &4109&0.453&0.513&0.664& 4096/4097\\
    mixed3-13&13 &8206&0.435&0.454&0.471& 8192/8193\\
    eco12&12 &1024&0.446&0.572&0.602& 1024/1024\\
    sot1&5 &8694&1.31&1.84&2.37& 1012/8694\\
    W1-6-5-2&5 &18000&0.462&0.471&0.472& 18000/18000\\
    W1-4-6-2&6 &6480&0.452&0.474&0.57& 6480/6480\\
    katsura10&11 &1024&0.557&0.661&0.652& 1024/1024
  \end{tabular}
\end{table}

The performance of $\mathsf{BlockParametrizationWithSplitting}$
depends on the density of $\mM_1$ and the number of points
$D_\mathfrak{A}$. For generic inputs, with no multiplicity,
$D_\mathfrak{A}=D$, so it actually will not spend any time
computing correction matrices or running $\mainalgoname{\sf
Residual}$. On the other hand, in the worst case, if
$D_\mathfrak{A}=0$, so that $D_\mathfrak{B}=D$, then Algorithm
{\sf BlockParametriza}{\sf -tionWithSplitting} may take more than twice as
long as $\mainalgoname$ (due to the two calls to respectively
$\mathsf{BlockParametrizationX}_1$ and
$\mathsf{BlockParametrizationResidual}$, together with the overhead induced 
by power projection).

This unlucky case was seldom seen in our experiments, since the
systems with multiplicities generated randomly had few multiple
points, and thus were favorable to us. An unfavorable case is system
``sot1'', where $V_\mathfrak{A}$ only accounts for $1012$ points out of
$7682$ points in the variety. 

\section*{Appendix}

In this appendix, we prove Theorem~\ref{coro:cost_approx} from
Section~\ref{section:matrix_seq}: {\em  Let $\seq = (\seqelt{s})_{s\ge 0}$ be a linearly recurrent sequence
  of matrices in $\sseqeltSpace$ and let $\degBd = \degBdl+\degBdr+1$,
  where $(\degBdl,\degBdr) \in \NN^2$ are such that the minimal left
  (resp.~right) matrix generators of $\seq$ have degree at most $\degBdl$
  (resp.~at most $\degBdr$). Then, given $\seqelt{0},\dots,\seqelt{d-1}$,
  one can compute a minimal left matrix generator of $\seq$ in
  $O(\rdim^\omega \M(\degBd) \log(\degBd))$ operations in $\field$.}
 The first lemma we need is similar
to \citep[Theorem~4.5]{Turner02}.
\begin{lemma}
  \label{lem:finitely_many_terms}
  Let $\seq = (\seqelt{s})_{s\ge 0}$ be a linearly recurrent sequence
  of matrices in $\sseqeltSpace$ and let $\degBdr \in \NN$ be such
  that minimal right matrix generators of $\seq$ have degree at most
  $\degBdr$.  Then, a vector $\rel =\row{p}_0 + \cdots
  +\row{p}_{\degBd}\var^\degBd \in \relSpace$ is a left relation for
  $\seq$ if and only if $\row{p}_0 \seqelt{s} + \cdots +
  \row{p}_{\degBd} \seqelt{s + \degBd} = \row{0}$ holds for $s \in
  \{0,\ldots,\degBdr-1\}$.
\end{lemma}
\begin{proof}
  In this proof, for a $u \times v$ polynomial matrix $\mat{Q}$, we denote by
  $\cdeg{\mat{Q}}$ the size-$v$ vector of the degrees of the columns of
  \(\mat{Q}\).  Consider a minimal right generator $\relbas \in
  \polMatSpace[\rdim]$ in Popov form (as defined e.g. in \citep{Kailath80}).
  Then, we have $\relbas = \mat{L}\,\diag{\var^{t_1},\ldots,\var^{t_{\rdim}}} -
  \mat{Q}$, where $\cdeg{\mat{Q}} < \cdeg{\relbas} = (t_1,\ldots,t_{\rdim})$
  termwise and $\mat{L} \in \matSpace[\rdim]$ is unit upper triangular. Define
  the matrix $\mat{U} =
  \diag{\var^{\degBdr-t_1},\ldots,\var^{\degBdr-t_{\rdim}}} \mat{L}^{-1}$,
  which is in $\polMatSpace[\rdim]$ since $\degBdr \ge \deg(\relbas) = \max_j
  t_j$. Then, the columns of the right multiple $\relbas \mat{U} =
  \var^{\degBdr} \mat{I}_\rdim - \mat{Q} \mat{U}$ are right relations for
  $\seq$, and we have $\deg(\mat{Q} \mat{U}) < \degBdr$. Thus, writing $\mat{Q}
  \mat{U} = \sum_{0 \le k < \degBdr} \mat{Q}_k \var^k$, we have
  $\seqelt{s+\degBdr} = \sum_{0 \le k < \degBdr} \seqelt{s+k} \mat{Q}_k$ for
  all $s \ge 0$.

  Assuming that $\row{p}_0 \seqelt{s} + \cdots + \row{p}_{\degBd} \seqelt{s +
  \degBd} = \row{0}$ holds for all $s \in \{0,\ldots,\degBdr-1\}$, we
  prove by induction that this holds for all $s\in\NN$. Let $s \ge
  \degBdr-1$ and assume that this identity holds for all integers up
  to $s$. Then, the identity concluding the previous paragraph implies
  that
  \begin{align*}
    \sum_{0 \le k \le \degBd} \row{p}_{k} \seqelt{s+1 + k} & =
    \sum_{0 \le k \le \degBd} \row{p}_{k} \left(\sum_{0\le j<\degBdr} \seqelt{s+1+k-\degBdr+j} \mat{Q}_j\right) 
                                                            = \sum_{0\le j<\degBdr} 
                                                           \underbrace{\left(\sum_{0 \le k \le \degBd} \row{p}_{k} \seqelt{s+1-\degBdr+j+k}\right)}_{=\, 0 \text{ since } s+1-\degBdr+j \le s} \mat{Q}_j = \row{0},
  \end{align*}
  and the proof is complete.
\end{proof}

The next result is similar to \citep[Theorem~4.6]{Turner02} (see also
Theorems~4.7 to~4.10 in that reference).  We recall from
\citep{BarBul92,BecLab94} that, given a matrix $\sys \in
\polMatSpace[\rdim][\rdim]$ and an integer $d \in \NN$, the set of
\emph{approximants for $\sys$ at order $d$} is defined as
\[
  \appMod{\sys}{d} = \{ \rel \in \relSpace \mid \rel \sys = \row{0} \bmod \var^d \}.
\]
Then, the next theorem shows that relations for $\seq$ can be retrieved as
subvectors of approximants at order about $\degBdl+\degBdr$ for a matrix
involving the first $\degBdl+\degBdr$ entries of $\seq$. 

\begin{theorem}
  \label{thm:mingen_via_appbas}
  Let $\seq = (\seqelt{s})_{s\ge 0}$ be a linearly recurrent sequence
  of matrices in $\sseqeltSpace$ and let $(\degBdl,\degBdr) \in \NN^2$
  be such that the minimal left (resp.~right) matrix generator of
  $\seq$ have degree at most $\degBdl$ (resp.~at most $\degBdr$).  For
  $\degBd>0$, define
  \[
    \sys =
    \begin{bmatrix}
      \sum_{0\le s < \degBd} \seqelt{s} \var^{\degBd-s-1} \\ - \mat{I}_{\rdim}
    \end{bmatrix} \in \polMatSpace[(\rdim+\rdim)][\rdim].
  \]
  Suppose that  $\degBd \ge \degBdr+1$ and let $\mat{B} \in \polMatSpace[(\rdim+\rdim)][(\rdim+\rdim)]$
  be a basis of $\appMod{\sys}{\degBdl+\degBdr+1}$. Then,
      if $\mat{B}$ is row reduced, it has exactly $\rdim$ rows of
      degree $\le\degBdl$, and they form a submatrix $[\relbas \;\; \remmat] \in
      \polMatSpace[\rdim][(\rdim+\rdim)]$ of $\mat{B}$ such that $\relbas$ is a
      minimal matrix generator for~$\seq$.
\end{theorem}
\begin{proof}
  We first observe that for any relation $\rel \in \relSpace$ for $\seq$, there exists $\rem \in
  \remSpace$ such that $\deg(\rem) < \deg(\rel)$ and $[\rel \;\; \rem]
  \in \appMod{\sys}{\degBd}$. Indeed, if $\rel$ is a relation for
  $\seq$ then $\num = \rel \seqpm$ has polynomial entries, where
  $\seqpm = \sum_{s\ge 0} \seqelt{s} \var^{-s-1}$. Then the vector
  $\rem = - \rel (\sum_{s \ge \degBd} \seqelt{s} \var^{\degBd-s-1})$
  has polynomial entries, has degree less than $\deg(\rel)$, and is
  such that $[\rel \;\; \rem] \sys = \num \var^{\degBd}$.

  Conversely, we show that for any vectors $\rel \in \relSpace$ and $\rem
  \in \remSpace$, if $[\rel \;\; \rem] \in\appMod{\sys}{\degBd}$ and
  $\deg([\rel \;\; \rem])\le\degBd-\degBdr-1$, then $\rel$ is a
  relation for $\seq$. Indeed, if $[\rel \;\; \rem]
  \in\appMod{\sys}{\degBd}$ we have $\rel (\sum_{0\le s< \degBd}
  \seqelt{s} \var^{\degBd-s-1}) = \rem \bmod \var^\degBd$. Since
  $\degBd\ge\degBdr+1$ and $\deg([\rel \;\;
  \rem])\le\degBd-\degBdr-1$, this implies that the coefficients of
  degree $\degBd-\degBdr$ to $\degBd-1$ of $\rel(\sum_{0\le s <
  \degBd} \seqelt{s} \var^{\degBd-s-1})$ are zero. Then,
  \cref{lem:finitely_many_terms} shows that $\rel$ is a relation for
  $\seq$. The theorem follows.
\end{proof}

Using the fast approximant basis algorithm of~\citep{GiJeVi03}, this
implies Theorem~\ref{coro:cost_approx}.

\bibliographystyle{elsarticle-harv}

\end{document}